\documentclass[11pt]{article}
\usepackage{graphicx,latexsym,amsmath,amsthm,amssymb,color,verbatim,mathrsfs,bbm,amscd,paralist}
\usepackage[bookmarks,colorlinks,breaklinks]{hyperref}  
\hypersetup{linkcolor=blue,citecolor=blue,filecolor=blue,urlcolor=blue} 

\usepackage[margin=2.5cm]{geometry}
\usepackage{mathtools}

\DeclarePairedDelimiter\floor{\lfloor}{\rfloor}
\usepackage{tikz}
\usetikzlibrary{arrows.meta}

\usetikzlibrary{decorations,calc}
\pgfdeclaredecoration{arrows}{draw}{
\state{draw}[width=\pgfdecoratedinputsegmentlength]{%
\path [every arrow subpath/.try] \pgfextra{%
    \pgfpathmoveto{\pgfpointdecoratedinputsegmentfirst}%
    \pgfpathlineto{\pgfpointdecoratedinputsegmentlast}%
   };
}}
\tikzset{every arrow subpath/.style={->, draw, thick}}

\newcommand{\IC}{\ensuremath{\mathbb{C}}}

\newcommand{\IQ}{\ensuremath{\mathbb{Q}}}
\newcommand{\IIF}{\ensuremath{\mathbb{F}}}
\newcommand{\IR}{\ensuremath{\mathbb{R}}}
\newcommand{\IN}{\ensuremath{\mathbb{N}}}
\newcommand{\per}{\mathrm{per}}
\newcommand{\GL}{\mathsf{GL}}
\newcommand{\End}{\mathsf{End}}
\newcommand{\GG}{\mathsf{G}}
\newcommand{\ag}{\mathfrak{g}}

\newcommand{\Ein}{\mathsf{in}}
\newcommand{\Eout}{\mathsf{out}}
\newcommand{\im}{\mathop{\textup{im}}}

\newcommand{\IMM}{f_{\textup{com}}}

\newcommand{\corr}{\textup{\textsf{corr}}}
\newcommand{\wt}{\mathsf{wt}}

\newcommand{\vi}{\text{\textcircled{$\cdot$}}}
\newcommand{\vstar}{\text{\textcircled{$\ast$}}}

\swapnumbers
\numberwithin{equation}{section}
\newtheorem{theorem}[equation]{Theorem}

\newtheorem{lemma}[equation]{Lemma}

\newtheorem{proposition}[equation]{Proposition}

\newtheorem{claim}[equation]{Claim}

\theoremstyle{definition}
\newtheorem{remark}[equation]{Remark}
\newtheorem{definition}[equation]{Definition}
\newtheorem{example}[equation]{Example}

\newcommand{\wW}{\ensuremath{\mathsf{w}}}
\newcommand{\nc}{\ensuremath{\mathsf{nc}}}
\newcommand{\m}{\ensuremath{\mathsf{m}}}

\newcommand{\IZ}{\ensuremath{\mathbb{Z}}}

\newcommand{\aS}{\ensuremath{\mathfrak{S}}}

\newcommand{\VNP}{\mathrm{VNP}}

\newcommand{\VP}{\mathrm{VP}}
\newcommand{\VQP}{\mathrm{VQP}}

\newcommand{\tr}{\ensuremath {\mathsf{tr}}}
\newcommand{\rk}{\ensuremath{\mathsf{rk}}}

\newcommand{\rat}{\ensuremath {\mathbb{Q}}}

\title{
\mbox{Algebraic Branching Programs, border complexity, and tangent spaces}
}

\author{
Markus Bl\"aser\thanks{Department of Computer Science, Saarland University, Saarland Informatics Campus, Saarbr\"ucken, Germany},
Christian Ikenmeyer\thanks{University of Liverpool. Part of this research was done when CI was at the Max Planck Institute for Software Systems, Saarbr\"ucken, Germany. CI was supported by DFG grant IK 116/2-1},
Meena Mahajan\thanks{The Institute of Mathematical Sciences, HBNI, Chennai, India},
Anurag Pandey\thanks{Max Planck Institute for Informatics, Saarland Informatics Campus, Saarbr\"ucken, Germany},
Nitin Saurabh\thanks{Technion-IIT, Haifa, Israel. Part of this work was done when the author was at the Max Planck Institute for Informatics, Saarbr\"ucken, Germany}}

\begin{document}
\raggedbottom

\maketitle

\begin{abstract}
Nisan showed in 1991 that the width of a smallest noncommutative single-(source,sink) algebraic branching program (ABP) to compute a noncommutative polynomial is given by the ranks of specific matrices. This means that the set of noncommutative polynomials with ABP width complexity at most $k$ is Zariski-closed, an important property in geometric complexity theory. It follows that approximations cannot help to reduce the required ABP width.

It was mentioned by Forbes that this result would probably break when going from single-(source,sink) ABPs to trace ABPs. We prove that this is correct. Moreover, we study the commutative monotone setting and prove a result similar to Nisan, but concerning the analytic closure. We observe the same behavior here: The set of polynomials with ABP width complexity at most $k$ is closed for single-(source,sink) ABPs and not closed for trace ABPs. The proofs reveal an intriguing connection between tangent spaces and the vector space of flows on the ABP.
We close with additional observations on VQP and the closure of VNP which allows us to establish a separation between the two classes.
\end{abstract}

\thispagestyle{empty}
\newpage
\setcounter{page}{1}

\section{Introduction and Results}\label{sec:intro}
Algebraic branching programs (ABPs) are an elegant model of computation that is widely studied in algebraic complexity theory (see e.g.\ \cite{bc:88, toda:92, mv:97, MP:08, aw:16, AFSSV:16, KNST:18, Kumar2019, FMST:19}) and is a focus of study in geometric complexity theory \cite{Lan:15, Ges:16, GIP:17}.
An ABP is a layered directed graph with $d+1$ layers of vertices (edges only go from layers $i$ to $i+1$)
such that the first and last layer have exactly the same number of vertices, so that each vertex in the first layer has exactly one so-called \emph{corresponding} vertex in the last layer.
One interesting classical case is when the first and last layer have exactly one vertex, which is usually studied in theoretical computer science.
We call this the \emph{single-(source,sink) model}.
Among algebraic geometers working on ABPs it is common to not impose restrictions on the number of vertices in the first and last layer \cite{Lan:15, Ges:16, Lan:17}.
We call this the \emph{trace model}.
Every edge in an ABP is labeled by a homogeneous linear form.
If we denote by $\ell(e)$ the homogeneous linear form of edge $e$, then we say that the ABP computes
$\sum_{p} \prod_{e \in p} \ell(e)$,
where the sum is over all paths that start in the first layer and end in the last layer at the vertex corresponding to the start vertex.

The \emph{width} of an ABP is the number of vertices in its largest layer.
We denote by $\wW(f)$ the minimal width required to compute $f$ in the trace model and we call $\wW(f)$ the \emph{trace ABP width complexity} of~$f$.
We denote by $\wW_1(f)$ the minimal width required to compute $f$ in the single-(source,sink) model and we call $\wW_1(f)$ the \emph{single-(source,sink) ABP width complexity} of~$f$.

The complexity class VBP is defined as the set of sequences of polynomials $(f_m)$ for which the sequence $\wW(f_m)$ is polynomially bounded.
Let $\per_m := \sum_{\pi \in \aS_m}\prod_{i=1}^m x_{i,\pi(i)}$ be the permanent polynomial. Valiant's famous $\text{VBP}\neq\text{VNP}$ conjecture can concisely be stated as
``The sequence of natural numbers $\big(\wW(\per_m)\big)_m$ is not polynomially bounded.'' Alternatively, this is phrased with $\wW_1$ or other polynomially related complexity measures in a completely analogous way.
In geometric complexity theory (GCT), one searches for lower bounds on algebraic complexity measures over $\IC$ such as $\wW$ and $\wW_1$ for explicit polynomials such as the permanent. All lower bounds methods in GCT and most lower bounds methods in algebraic complexity theory are \emph{continuous}, which means that if $f_\varepsilon$ is a curve of polynomials with $\lim_{\varepsilon \to 0} f_\varepsilon = f$ (coefficient-wise limit) and $\wW(f_\varepsilon) \leq w$, then these methods cannot be used to prove $\wW(f) > w$. This is usually phrased in terms of so-called \emph{border complexity} (see e.g.~\cite{BLMW:11,Lan:15}):
The \emph{border trace ABP width complexity} $\underline{\wW}(f)$ is the smallest $w$ such that $f$ can be approximated arbitrarily closely by polynomials $f_\varepsilon$ with $\wW(f_\varepsilon) \leq w$.
Analogously, we define the \emph{border single-(source,sink) ABP width complexity} $\underline{\wW_1}(f)$ as the smallest $w$ such that $f$ can be approximated arbitrarily closely by polynomials $f_\varepsilon$ with $\wW_1(f_\varepsilon) \leq w$.
Analogously to VBP we define $\overline{\text{VBP}}$ as the set of sequences of polynomials whose $(\underline{\wW}(f_m))$ is polynomially bounded.
Clearly $\text{VBP} \subseteq \overline{\text{VBP}}$.
Mulmuley and Sohoni \cite{gct1, gct2, BLMW:11} (see also \cite{bue:01} for a related conjecture) conjectured a strengthening of Valiant's conjecture, namely that $\text{VNP} \not\subseteq \overline{\text{VBP}}$.
In principle it could be that $\underline{\wW}(f) < \wW(f)$; the gap could even be superpolynomial, which would mean that $\text{VBP} \subsetneq \overline{\text{VBP}}$.
If $\text{VBP} = \overline{\text{VBP}}$, then Valiant's conjecture is the same as the Mulmuley-Sohoni conjecture, which would mean that if $\text{VBP}\neq\text{VNP}$, then continuous lower bounds methods exist that show this separation.

Border complexity is an old area of study in algebraic geometry. In theoretical computer science it was introduced by Bini et al.~\cite{bcrl:79}, where \cite{Bini:80} proves that in the study of fast matrix multiplication, the gap between complexity and border complexity is not too large.
The study of the gap between complexity and border complexity of algebraic complexity measures in general has started recently \cite{GMQ:16, BIZ:18, Kumar2018} as an approach to understand if strong algebraic complexity lower bounds can be obtained from continuous methods.

In this paper we study two very different settings of ABPs: The noncommutative and the monotone setting.
To capture commutative, noncommutative, and monotone computation, let $R$ be a graded semiring with homogeneous components $R_d$.
In our case the settings for $R_d$ are
\begin{itemize}
\item $R_d=\IIF[x_1,\ldots,x_m]_d$ the set of homogeneous degree $d$ polynomials in $m$ variables over a field~$\IIF$,
\item $R_d=\IIF\langle x_1,\ldots,x_m\rangle_d$ the set of homogeneous degree $d$ polynomials in $m$ noncommuting variables over a field~$\IIF$,
\item $R_d=\IR_{+}[x_1,\ldots,x_m]_d$ the set of homogeneous degree $d$ polynomials in $m$ variables with nonnegative coefficients.
\end{itemize}
As it is common in the theoretical computer science literature, we call elements of $R_d$ \emph{polynomials}. Note that $\IIF\langle x_1,\ldots,x_m\rangle_d$ is naturally isomorphic to the $d$-th tensor power of $\IIF^m$, so \emph{tensor} would be the better name. We hope that no confusion arises when in the later sections (where we use concepts from multilinear algebra) we use the tensor language.
In the \emph{homogeneous setting}, all ABP edge labels are in $R_1$, and hence the polynomial that is computed is in $R_d$.
In the \emph{affine setting}, all ABP edge labels are in $R_0+R_1$, and hence the polynomial that is computed is in $\bigoplus_{d'\leq d}R_{d'}$.

\subsection*{Noncommutative ABPs}
Let $R_d=\IIF\langle x_1,\ldots,x_m\rangle_d$ and consider the homogeneous setting. We write $\nc\wW$ instead of $\wW$ and $\nc\wW_1$ instead of $\wW_1$ to highlight that we are in the noncommutative setting.
Nisan \cite{nisan1991lower} proved:
\begin{theorem}\label{thm:nisan}
Let $M_i$ denote the $n^i \times n^{d-i}$ matrix whose entry at position $((k_1,\ldots,k_i),(k_{i+1},\ldots,k_d))$ is the coefficient of the monomial $x_{k_1}x_{k_2}\cdots x_{k_d}$ in $f$.
Then every single-(source,sink) ABP computing $f$ has at least $\rk(M_i)$ many vertices in layer $i$. Conversely, there exists a single-(source,sink) ABP computing $f$ with exactly $\rk(M_i)$ many vertices in layer $i$.
\end{theorem}
Nisan used this formulation to prove strong complexity lower bounds for the noncommutative determinant and permanent.
Forbes \cite{Forbes16video} remarked that Theorem~\ref{thm:nisan} implies that for fixed $w$
\begin{equation}\label{eq:Zariskiclosed}
\text{the set } \{f \mid \nc\wW_1(f) \leq w \} \text{ is Zariski-closed}
\end{equation}
and hence that
\begin{equation}\label{eq:nisan}
\underline{\nc\wW_1}(f) = \nc\wW_1(f) \text{ for all $f$}.
\end{equation}

Proving a similar result (even up to polynomial blowups) in the commutative setting would be spectacular: It would imply $\text{VBP}=\overline{\text{VBP}}$ and hence that Valiant's conjecture is the same as the Mulmuley-Sohoni conjecture.
By a general principle, for all standard algebraic complexity measures, over $\IC$ we have that the Zariski-closure of a set of polynomials of complexity at most $w$ equals the Euclidean closure \cite[\S2.C]{mum:complprojvars}.

Forbes mentioned that he believes that Nisan's proof cannot be lifted to the trace model. In this paper we prove that Forbes is correct, by constructing a polynomial $f_0$ with
\begin{equation}\label{eq:counterexample}
\underline{\nc\wW}(f_0) < \nc\wW(f_0).
\end{equation}

The proof is given in Sections~\ref{sec:f0-construction}--\ref{sec:flows}. It is a surprisingly subtle application of differential geometry (inspired by \cite{HL:16})
and interprets tangent spaces to certain varieties as vector spaces of flows on an ABP digraph.

The gap between $\underline{\nc\wW}(f)$ and ${\nc\wW}(f)$ can never be very large though:

\begin{samepage}
\begin{equation}\label{eq:sandwich}
\underline{\nc\wW}(f) \leq \nc\wW(f) \leq \nc\wW_1(f)
\stackrel{\eqref{eq:nisan}}{=} \underline{\nc\wW_1}(f)
\stackrel{\footnotemark}{\leq} 
\big(\underline{\nc\wW}(f)\big)^2 \text{ for all $f$}.
\end{equation}
\footnotetext{Given a trace ABP $\Gamma$ computing $f$ and a pair of corresponding start and end vertices, we can extract a single-(source,sink) ABP by deleting all other start and end vertices. If we do this for each pair of start and end vertices, and if we then idenfity all start vertices to a single start vertex, and also all end vertex to a single end vertex, then we obtain a single-(source,sink) ABP computing $f$. The width has grown by a factor of $w$, where $w$ is the number of start vertices in $\Gamma$.}
\end{samepage}

It is worth noting that for our separating polynomial $f_0$, the gap is even less; $\underline{\nc\wW}(f_0)<\nc\wW(f_0) \leq 2\underline{\nc\wW}(f_0)$.
This is the first algebraic model of computation where complexity and border complexity differ, but their gap is known to be polynomially bounded!
For most models of computation almost nothing is known about the gap between complexity and border complexity. For commutative width 2 affine ABPs the gap is even as large as between computable and non-computable \cite{BIZ:18}!

\subsection*{Monotone ABPs}
Let $R_d=\IR_{+}[x_1,\ldots,x_m]_d$ and consider the affine or homogeneous setting.

Since $\IR$ is not algebraically closed, we switch to a more algebraic definition of approximation.
Let $\IR[\varepsilon,\varepsilon^{-1}]_{+}$ denote the ring of Laurent polynomials that are nonnegative for all sufficiently small~$\varepsilon>0$. Clearly, elements from $\IR[\varepsilon,\varepsilon^{-1}]_{+}$ can have a pole at $\varepsilon=0$ of arbitrarily high order.
We define $\underline{\m\wW}(f)$ to be the smallest $w$ such that there exists a polynomial $f'$ over the ring $\IR[\varepsilon,\varepsilon^{-1}]_{+}$ such that
\begin{itemize}
 \item there exists a width $w$ ABP over $\IR[\varepsilon,\varepsilon^{-1}]_{+}$ that computes $f'$,
 \item no coefficient in $f'$ contains an $\varepsilon$ with negative exponent, and setting $\varepsilon$ to 0 in $f'$ yields $f$, i.e., $f'_{\varepsilon=0} = f$.
\end{itemize}

We prove a result that is comparable to \eqref{eq:nisan}, but uses a very different proof technique:
\begin{equation}\label{eq:nisanmonotone}
\underline{\m\wW_1}(f) = \m\wW_1(f) \text{ for all $f$}.
\end{equation}
In terms of complexity classes this can concisely be written as
\[
\text{MVBP}=\overline{\text{MVBP}}^{\IR}.
\]
Our proof also works if the ABP is not layered and the labels are affine.

Intuitively, in this monotone setting, one would think that approximations do not help, because there cannot be cancellations. But quite surprisingly the same construction as in \eqref{eq:counterexample} can be used to find $f_0$ such that
\begin{equation}\label{eq:counterexamplemonotone}
\underline{\m\wW}(f_0) < \m\wW(f_0).
\end{equation}
By the same reasoning as in \eqref{eq:sandwich}, we obtain
\begin{equation}
\underline{\m\wW}(f) \leq \m\wW(f) \leq \big(\underline{\m\wW}(f)\big)^2 \text{ for all $f$}.
\end{equation}
This gives a natural monotone model of computation where approximations speed up the computation. Again, the gap is polynomially bounded!

\subsection*{Separating \texorpdfstring{$\VQP$}{VQP} from \texorpdfstring{$\overline{\VNP}$}{VNPbar}}
B\"urgisser in his monograph \cite{Burgisser2000} defined the complexity class $\VQP$ as the class of polynomials with quasi-polynomially bounded straight-line programs,  and established its relation to the classes $\VP$ and $\VNP$ (see Section \ref{sec:vqp} for definitions). He showed that the determinant polynomial is $\VQP$-complete with respect to the so-called $qp$-projections (see \cite{Burgisser2000}, Corollary 2.29). He strengthened Valiant's hypothesis of $\VNP \not\subseteq \VP$ to $\VNP \not\subseteq \VQP$ and called it \emph{Valiant's extended hypothesis} (see \cite{Burgisser2000}, section 2.5).
He further showed that $\VP$ is strictly contained in $\VQP$ as one would intuitively expect (see \cite{Burgisser2000}, section 8.2).
Finally, he also showed that $\VQP$ is not contained in $\VNP$ (see \cite{Burgisser2000}, Proposition 8.5 and Corollary 8.9). In this article, we observe that his proof is stronger and actually shows that $\VQP$ is not contained in $\overline{\VNP}$ either, where $\overline{\VNP}$ is the closure of the complexity class $\VNP$ in the sense as mentioned above.

\subsection*{Structure of the paper}

In Section~\ref{sec:monotoneclosed} we prove \eqref{eq:nisanmonotone}.
Sections~\ref{sec:f0-construction} to~\ref{sec:flows} are dedicated to proving \eqref{eq:counterexample} and \eqref{eq:counterexamplemonotone} via a new connection between tangent spaces and flow vector spaces.
In Section~\ref{sec:vqp}, we discuss the separation between $\mathrm{VQP}$ and $ \overline{\mathrm{VNP}}$.

\section{Related work}
Grenet \cite{grenet:11} showed that $\m\wW(\per_m) \leq \binom{m}{\lceil m/2\rceil}$ by an explicit construction of a monotone single-(source,sink) ABP.
Even though the construction is monotone, its size is optimal for $m=3$ \cite{abv:15} (for 4 this is already unknown).
The noncommutative version of this setting has been studied in \cite{FMST:19}.
\cite{Yeh:19} recently showed that the monotone circuit classes MVP and MVNP are different. We refer the reader to \cite{Yeh:19} and \cite{Sri:19} and the references therein to get more information about monotone algebraic models of computation and their long history.

\cite{HL:16} present a method that can be used to show that a complexity measure and its border variant are not the same. They used it to prove that an explicit polynomial has border determinantal complexity 3, but higher determinantal complexity. We use their ideas as a starting point in Section~\ref{sec:f0-construction} and the later sections.

\section{Preliminaries}

For a homogeneous degree $d$ ABP $\Gamma$,
we denote by $V$ the set of vertices of $\Gamma$ and by $V^i$ the set of vertices in layer $i$, $1 \leq i \leq d+1$.
We choose an explicit bijection between the sets $V^1$ and $V^{d+1}$, so that each vertex $v$ in $V^1$ has exactly one \emph{corresponding} vertex $\corr(v)$ in $V^{d+1}$.
We denote by $E^i$ the set of edges from $V^i$ to $V^{i+1}$. Let $E$ denote the union of all $E^i$.

There is a classical interpretation in terms of iterated matrix multiplication:
Fix some arbitrary ordering of the vertices within each layer, such that the $i$-th vertex in $V^1$ corresponds to the $i$-th vertex in $V^{d+1}$.
For $1\leq k \leq d$ let $M_k$ be the $|V^k|\times|V^{k+1}|$ matrix whose
entry at position $(i,j)$ in $M_k$ is the label from the $i$-th vertex in $V^k$ to the $j$-th vertex in $V^{k+1}$.
Then $\Gamma$ computes the trace
\begin{equation}\label{eq:trace}
\sum_{\substack{1\leq k_1 \leq |V^1| \\ 1\leq k_2 \leq |V^2| \\ \vdots \\ 1\leq k_d \leq |V^d|}} (M_1)_{k_1,k_2} (M_2)_{k_2,k_3} \cdots (M_{d-1})_{k_{d-1},k_d} (M_d)_{k_d,k_1} = \tr\big(M_1 M_2 \cdots M_d \big).
\end{equation}
Hence the name \emph{trace model}. In the single-(source,sink) model,
the trace is taken of a $1 \times 1$ matrix.

\section{Monotone commutative single-(source,sink) ABPs are closed}\label{sec:monotoneclosed}
For fixed $w \in \IN$ we study
\begin{equation}\label{eq:Zariskiclosedmonotone}
\text{the set } \{f \in  \IR_+[x_1,\ldots,x_n]_d \mid \m\wW_1(f) \leq w \}.
\end{equation}
We first start with the simple observation that it is \emph{not} Zariski-closed.
\begin{proposition}
$\{f \in \IR_+[x_1,\ldots,x_n]_d \mid \m\wW_1(f) \leq w \}$ is not Zariski-closed.
\end{proposition}
\begin{proof}
An analogous statement is true for all natural algebraic complexity measures.
Note that a homogeneous degree $d$ single-(source,sink) width $w$ ABP has $2w+w^2(d-2)$ many edges. The label on each edge is a linear form in $n$ variables, so such an ABP is determined by $N := n(2w+w^2(d-2))$ many parameters. Let $F : \IC^N \to \IC[x_1,\ldots,x_n]_d$ be the map that maps these parameters to the polynomial computed by the ABP. Every coordinate function of $F$ is given by polynomials in $N$ variables, so $F$ is Zariski-continuous. Therefore
\[
\overline{F((\IR_+)^N)} = \overline{F(\overline{(\IR_+)^N})} = \overline{F(\IC^N)} \supseteq F(\IC^N) \supsetneqq F((\IR_+)^N), \]
where the overline means the Zariski closure.
\end{proof}
Recall that an ABP has $d+1$ layers of vertices. 
If an ABP has $w_i$ many vertices in layer $i$, $1 \leq i \leq d$, we say the ABP has \emph{format} $w = (w_1,w_2,\ldots,w_d)$. We further recall that $w_{d+1} = w_1$.
The following theorem is our closure result, which proves \eqref{eq:nisanmonotone} and hence $\text{MVBP}=\overline{\text{MVBP}}^{\IR}$.
\begin{theorem}\label{thm:monotoneclosed}
Given a polynomial $f$ over $\IR$ and given a format $w$ single-(source,sink) ABP with affine linear labels over $\IR[\varepsilon,\varepsilon^{-1}]_+$ computing $f_\varepsilon$ such that $\lim_{\varepsilon\to 0}f_\varepsilon=f$.
Then there exists a format $w$ monotone single-(source,sink) ABP that computes $f$.
\end{theorem}
\begin{proof}
The proof is constructive and done by a two-step process.
In the first step (which is fairly standard and works in many computational models) we move all the $\varepsilon$ with negative exponents to edges adjacent to the source. The second step then uses the monotonicity.

Given $\Gamma$ with affine linear labels over $\IR[\varepsilon,\varepsilon^{-1}]_+$ we repeat the following process until
all labels that contain an $\varepsilon$ with a negative exponent are incident to the source vertex.
\begin{compactitem}
 \item Let $e$ be an edge whose label contains $\varepsilon$ with a negative exponent $-i<0$. Moreover, assume that $e$ is not incident to the source vertex. Let $v$ be the start vertex of~$e$. We rescale all edges outgoing of $v$ with $\varepsilon^{i}$ and we rescale all edges incoming to $v$ with $\varepsilon^{-i}$.
\end{compactitem}
If we always choose the edge with the highest layer, then it is easy to see that this process terminates.
Since every path from the source to the sink that goes through a vertex $v$ must use exactly one edge that goes into $v$ and exactly one edge that comes out of $v$, throughout the process the value of $\Gamma$ does not change.
We finish this first phase by taking the highest negative power $i$ among all labels of edges that are incident to the source and then rescale all these edges with $\varepsilon^i$.
The resulting ABP $\Gamma^i$ computes $\varepsilon^i f$ and no label contains an $\varepsilon$ with negative exponent.
We now start phase 2
that transforms $\Gamma^i$ into $\Gamma^{i-1}$ that computes $\varepsilon^{i-1} f$ without introducing negative exponents of $\varepsilon$.
We repeat phase 2 until we reach $\Gamma^0$ in which we safely set $\varepsilon$ to 0.
Throughout the whole process we do not change the structure of the ABP and only rescale edge labels with powers of $\varepsilon$, which preserves monotonicity, so the proof is finished.
It remains to show how $\Gamma^i$ can be transformed into $\Gamma^{i-1}$.
An edge whose label is divisible by $\varepsilon$ is called an \emph{$\varepsilon$-edge}.
Consider the 
set $\Delta$ of vertices that are reachable from the source using only
non $\varepsilon$-edges in $\Gamma^i$.
The crucial insight is that since $\Gamma^i$ is monotone and computes a polynomial that is divisible by $\varepsilon$, we know that every path in $\Gamma^i$ from the source to the sink uses an $\varepsilon$-edge.
Therefore $\Delta$ cannot contain the sink.
We call a vertex in $\Delta$ whose outdegree is zero a \emph{leaf} vertex. We repeat the following procedure until the source is the only leaf vertex.
\begin{compactitem}
\item 
  Let $v$ be a non-source leaf vertex in $\Delta$. We rescale all edges outgoing of $v$ with $\varepsilon^{-1}$ and we rescale all edges incoming to $v$ with $\varepsilon$.
\end{compactitem}
It is easy to see that this process terminates with the source being the only leaf vertex.
Since the source is a leaf vertex, all edges incident to the source are $\varepsilon$-edges. We divide all their labels by $\varepsilon$ to obtain $\Gamma^{i-1}$.
\end{proof}

\section{Explicit construction of $f_0$ with higher complexity than border complexity}
\label{sec:f0-construction}
Fix some $d\geq 3$. In this section for every $m \geq 2$ we construct $f_0$ such that
\begin{equation}\label{eq:goal}
m = \underline{\nc\wW}(f_0)<\nc\wW(f_0). 
\end{equation}
A completely analogous construction can be used to find $f_0$ with $\underline{\wW}(f_0)<\wW(f_0)$ and with $\underline{\m\wW}(f_0)<\m\wW(f_0)$. For the sake of simplicity, we carry out only the proof for \eqref{eq:goal}.

Recall that in a format $w$ ABP we have $w_{d+1} = w_1$.
In each layer $i$ we enumerate the vertices $V^i = \{v^i_1,\ldots, v^i_{w_i}\}$ and we assume without loss of generality that the correspondence bijection between $V^{d+1}$ and $V^{1}$
is the identity on the indices $j$ of $v^1_{j}$, i.e., the $j$th vertex in $V^1$ corresponds to the $j$th vertex in $V^{d+1}$.

Fix an ABP format $w = (w_1,w_2,\ldots,w_d)$ such that for all $i$, $w_i \geq 2$.
Let $\Gamma_{\textup{com}}$ denote the directed acyclic graph underlying an ABP of format $w$.
An edge can be described by the triple $(a,b,i)$, where $1 \leq i \leq d$, $1 \leq a \leq w_i$ and $1 \leq b \leq w_{i+1}$. Consider the following labeling of the edges with triple-indexed variables:
$\ell_{\textup{com}}((a,b,i)) = x^{(i)}_{(a,b)}$. Define $\IMM$ to be the polynomial computed by $\Gamma_{\textup{com}}$ with edge labels $\ell_{\textup{com}}$.

We now construct $f_0$ as follows. Let $d$ be odd (the case when $d$ is even works analogously).
Since in each layer we enumerated the vertices, we can now assign to each vertex its parity: even or odd.
We call an edge between two even or two odd vertices \emph{parity preserving}, while we call the other edges \emph{parity changing}. Let us consider the following labeling of $\Gamma_{\textup{com}}$:
We set $\ell_0((a,b,i)) := x^{(i)}_{(a,b)}$ if $(a,b,i)$ is parity changing (i.e., $a \not\equiv b \pmod 2$)
and set the label
$\ell_0((a,b,i)) := \varepsilon x^{(i)}_{(a,b)}$ otherwise, where $\varepsilon \in \IC$. 
Let $f'_\varepsilon$ be the polynomial computed by $\Gamma_{\textup{com}}$ with edge labels $\ell_0$ and set
$f_\varepsilon := \frac 1 \varepsilon f'_\varepsilon$ for $\varepsilon\neq 0$.
We define $f_0 := \lim_{\varepsilon\to 0} f_\varepsilon$ (convergence follows from the construction, because $d$ is odd).
By definition, for all $\varepsilon \neq 0$, $f_\varepsilon$ can be computed by a format $w$ ABP. However, we will now prove that this property fails for the limit point~$f_0$.

\begin{theorem}\label{thm:separation}
Fix an ABP format $w = (w_1,w_2,\ldots,w_d)$ such that for all $i$, $w_i \geq 2$.
Let $f_0$ be defined as above. Then, $f_0$ \emph{cannot} be computed by an ABP of format $w$.
\end{theorem}
Note that for a format where $m=w_1=\cdots=w_d$, this gives the $f_0$ which was desired in \eqref{eq:goal}.
(Note, however, that $f_0$ can be computed by an ABP of width $2m$ as follows.
Construct an ABP $\Gamma'$ that has, for each vertex $v\in \Gamma_{\textup{com}}$,
vertices $v'$ and $v''$. For each parity changing edge
$(a,b) \in \Gamma_{\textup{com}}$ with label $\ell_0$, add edges $(a',b')$ and
$(a'',b'')$ with the same label $\ell_0$. For each parity preserving
edge $(a,b) \in \Gamma_{\textup{com}}$ with label $\ell_0$, add edge
$(a',b'')$ with label $(\frac 1 \varepsilon)\ell_0$. For corresponding
vertices $u,v$ in $\Gamma_{\textup{com}}$, let $v''$ be the
corresponding vertex for $u'$ and $v'$ be the corresponding vertex for
$u''$ in $\Gamma'$. 
All paths between corresponding vertices in this ABP use exactly one
parity preserving edge of $\Gamma_{\textup{com}}$, and so this ABP computes $f_0$.)

The proof of Theorem~\ref{thm:separation} works as follows.
Let $\GG := \GL_{w_1 w_2} \times \GL_{w_2 w_3} \times \cdots \times \GL_{w_d w_{d+1}}$.
Let $\End := \overline{G}$ denote its Euclidean closure, i.e., tuples of matrices in which one or several matrices can be singular.

We consider noncommutative homogeneous polynomials in the variables $x_{(a,b)}^{(i)}$ such that the $i$-th variable in each monomial is $x_{(a,b)}^{(i)}$ for some $a \in [w_i]$ and $ b \in [w_{i+1}]$. The vector space of these polynomials is isomorphic to
$W := \IC^{w_1 w_2} \otimes \IC^{w_2 w_3} \otimes \cdots \otimes \IC^{w_d w_{d+1}}$ and the monoid $\End$ (and thus also the group $\GG$) acts on this space in the canonical way.
The set
\[
\{f \in W \mid f \text{ can be computed by a format $w$ ABP}\}
\]
is precisely the orbit $\End \IMM$.
We follow the overall proof strategy in \cite{HL:16}. The monoid orbit $\End \IMM$ decomposes into two disjoint orbits:
\[
\End \IMM = \GG \IMM \cup (\End \setminus \GG) \IMM.
\]
Our goal is to show two things independently:
\begin{enumerate}
\item $f_0 \notin (\End \setminus \GG) \IMM$, and
\item $f_0 \notin \GG \IMM$,
\end{enumerate}
which finishes the proof of Theorem~\ref{thm:separation}.

All elements in $(\End \setminus \GG) \IMM$ are \emph{not concise},
a term that we define in Section \ref{sec:concise}, where we also prove that $f_0$ is concise. Therefore $f_0 \notin (\End \setminus \GG) \IMM$.

All elements in $\GG \IMM$ have \emph{full orbit dimension},
a term that we define in Section \ref{sec:orbitdimension} 
and we prove that $f_0$ does \emph{not} have full orbit dimension in Section \ref{sec:flows}.
This finishes the proof of Theorem~\ref{thm:separation}.

%

\section{Conciseness}\label{sec:concise}

In this section we show that $f_0 \notin (\End \setminus \GG) \IMM$. To do so
we use a notion called \emph{conciseness}. Informally, it captures whether
a polynomial depends on all variables independent of a change of basis, or a tensor cannot be embedded into a tensor product of smaller spaces.

Given a tensor $f$ in $\IC^{m_1}\otimes \IC^{m_2} \otimes \cdots \otimes \IC^{m_d}$, we associate the following matrices with $f$.
For $j \in[d]$, define a matrix $M_f^j$ of dimension
$m_j \times (\prod_{i \in [d]\setminus \{j\}}m_i)$ with rows labeled by the standard
basis of $\IC^{m_j}$, and columns by elements in the Cartesian product
$\{\text{standard basis of } \IC^{m_1}\} \times \cdots \times \{\text{standard basis of } \IC^{m_{j-1}}\} \times \{\text{standard basis of } \IC^{m_{j+1}}\}\times \cdots \times \{\text{standard basis of } \IC^{m_d}\}$.
We write the tensor $f$ in the standard basis
\[
f = \sum_{\substack{1 \leq i_1 \leq m_1\\1 \leq i_2 \leq m_2\\\vdots\\1 \leq i_d \leq m_d}} \alpha_{i_1,\ldots,i_d} e_{i_1}\otimes \cdots \otimes e_{i_d}
\]
and associate to it the matrix $M_f^j$
whose entry at position $((i_j),(i_1,i_2,\ldots,i_{j-1},i_{j+1},\ldots,i_d))$ is $\alpha_{i_1,\ldots,i_d}$.

\begin{definition}
  \label{def:concise}
  We say that a tensor $f$ in $\IC^{m_1}\otimes \IC^{m_2} \otimes \cdots \otimes \IC^{m_d}$ is \emph{concise} if and only if for all $j \in [d]$, $M_f^j$ has full rank. 
\end{definition}


As a warm-up exercise we now show that $\IMM$ is concise.

\begin{proposition}
  \label{prop:fcom-concise}
  $\IMM$ is concise. 
\end{proposition}
\begin{proof}
  We know that $\IMM \in W$. Let us consider the matrix $M_{\IMM}^j$ for some
  $j \in [d]$. To establish that $M_{\IMM}^j$ has full rank, it suffices to
  show that rows are linearly independent. In order to show that, we argue
  that every row is non-zero and every column has at most one non-zero entry.
  In other words, rows are supported on disjoint sets of columns.  

  A row of $M_{\IMM}^j$ is labeled by an edge in the $j$-th layer of the ABP
  $\Gamma_{\textup{com}}$. Recall that only paths that start at
  a vertex in $V^1$ and end at the corresponding vertex in $V^{d+1}$ contribute
  to the computation in $\Gamma_{\textup{com}}$. We call such paths
  \emph{valid paths}. An entry in $M_{\IMM}^j$ is non-zero iff the corresponding
  row and column labels form a valid path in $\Gamma_{\textup{com}}$. Thus, it is
  easily seen that a row is non-zero iff there is a valid path in
  $\Gamma_{\textup{com}}$ that \emph{passes through} the edge given by the row
  label. By the structure of $\Gamma_{\textup{com}}$, in particular that every
  layer is a complete bipartite graph, we observe that passing through every edge there is some valid path.
  Hence, we obtain that every row is non-zero.

  The second claim now follows from the observation that fixing $d-1$ edges either defines a unique $d$th edge so that these $d$ edges form a valid path,
  or for these $d-1$ edges there is no such $d$th edge.
\end{proof}

As mentioned in Section~\ref{sec:f0-construction}, to establish
$f_0 \notin (\End \setminus \GG) \IMM$ we will show that $f_0$ is concise while
any element in $(\End \setminus \GG) \IMM$ is not. 

\begin{lemma}
  \label{lem:f0-concise}
  $f_0$ is concise. 
\end{lemma}
\begin{proof}
  Analogous to the proof of Proposition~\ref{prop:fcom-concise}, we again show
  that every row of $M_{f_0}^j$ is non-zero and every column of it has at most
  one non-zero entry. That is, rows of $M_{f_0}^j$ are supported on
  disjoint sets of columns.
  
  From the construction of $f_0$ it is seen that a path in $\Gamma_{\textup{com}}$
  contributes to the computation of $f_0$ iff it is a valid path that comprises
  of \emph{exactly one} parity preserving edge. The second claim of every
  column having at most one non-zero entry now follows for the same reason as
  in the proof of Proposition~\ref{prop:fcom-concise}.
  
  Before proving the first claim, we recall two assumptions in the construction
  of $f_0$. The first is that the format  $w = (w_1,w_2,\ldots,w_d)$ is such
  that $w_i \geq 2$ for all $i \in[d]$ and the second is that $d$ is odd. To argue
  that a row is non-zero it suffices to show that a valid path comprising of
  only one parity preserving edge passes through the edge given by the row
  level. Let us consider an arbitrary edge $e$ in $\Gamma_{\textup{com}}$.
  We have two cases to consider depending on whether it is parity
  \emph{preserving} or \emph{changing}.

  \textbf{Case 1.} Suppose $e$ is parity preserving and it belongs to a layer
  $j \in [d]$. The number of layers on the left of $e$ is $j-1$ and on the
  right is $d-j$. Since $d$ is odd, these numbers are either both even or both odd.
  We now argue for the case when they are even (the odd case is analogous).
  Choose a vertex $v$ in $V^1$ that has the same parity
  (different in the odd case) as one of the end points of $e$.
  (Such a choice exists because $w_1 \geq 2$.) We now claim that there exists
  a valid path starting at $v$ that passes through $e$ and contains exactly one
  parity preserving edge. Since $e$ is parity preserving, all edges in the
  claimed path must be parity changing. We observe that $e$ can be easily
  extended in both directions using parity changing edges such that the path
  ends at $\corr(v)$. The existence of parity changing edges at each layer uses the
  assumption that $w_i \geq 2$.

  \textbf{Case 2.} Otherwise $e =(a,b)$ is parity changing. Again as before
  there are two cases based on whether both $j-1$ and $d-j$ are even or odd.
  Consider the case when they are even (the odd case being analogous).
  We first assume that $j \neq d$.
  Choose
  a vertex $v$ in $V^1$ that has the same parity as $a$. We now construct a
  valid path from $v$ to $\corr(v)$ that passes through $e$ and contains exactly one
  parity preserving edge. 
  It is easily seen that there exists a path from $v$ to $a$ using only
  parity changing edges. We choose a parity preserving outgoing edge incident
  to $b$. We call its endpoint $v_1$. Since $v_1$ and $v$ have different parities,
  we can connect $v_1$ to $\corr(v)$ in $V^{d+1}$ using only parity changing edges.
  Thus we obtain the following valid path $v \to\cdots \to a\to b\to v_1 \to\cdots \to\corr(v)$ passing
  through exactly one parity preserving edge $(b,v_1)$.
  In the case that $j = d$, choose an incoming parity preserving edge incident on $a$ instead of an outgoing edge on $b$.
\end{proof}

%
%

\begin{remark}
  \label{rem:proof-fails}
  We note that if the format $w = (w_1,\ldots ,w_d)$ defining $f_0$ is
  such that for some $j \in [d]$, $w_j = 1$, then $f_0$ is \emph{not} concise.
  This can be seen as follows.

  Let $w_j=1$, and let $v$ denote the unique vertex in $V^j$.  Let $e$
  be the edge $e=(1,1,j)$.  If $j < d$, let $e'$ be the edge
  $e'=(1,1,j+1)$, otherwise let $e'$ be the edge $e'=(1,1,j-1)$. Both
  $e,e'$ are parity preserving edges. By construction, every valid path
  using $e'$ must also use $e$. Hence the corresponding row in the
  matrix $M_{f_0}^{j+1}$ if $j< d$, and in $M_{f_0}^{j-1}$ otherwise,
  is zero.  Therefore $f_0$ is not concise.
  
  This is an interesting observation, because this is
  the point where our proof fails for single-(source,sink) ABPs,
  and this is expected, because Nisan~\cite{nisan1991lower} had shown
  that the set of
  polynomials computed by such ABPs of format $w$ is a closed set.
\end{remark}

\begin{lemma}
  \label{lem:non-concise-orbit}
  Let $f \in (\End \setminus \GG) \IMM$. Then $f$ is not concise. 
\end{lemma}
\begin{proof}
This statement is true in very high generality. In our specific case a proof goes as follows.
If $f \in (\End \setminus \GG) \IMM$, then $f = g \IMM$ for some $g \in \End\setminus\GG$.
Let $g = (g_1,\ldots,g_d)$, where $g_i \in \IC^{w_i w_{i+1}\times w_i w_{i+1}}$.
Since $g \notin \GG$, at least one of the $g_i$ must be singular.
The crucial property is $M_{g\IMM}^i = g_i M_{\IMM}^i$,
which finishes the proof.
\end{proof}

\section{Orbit dimension, tangent spaces, and flows}
\label{sec:orbitdimension}

In this section we introduce tangent spaces and study their dimensions. We especially study them in the context of $\GG \IMM$, and $\GG f_0$.

The \emph{orbit dimension} of a tensor
$f \in \IC^{w_1 w_2} \otimes \IC^{w_2 w_3} \otimes \cdots \otimes \IC^{w_d w_{d+1}}$
is the dimension of the orbit $\GG f$ as an affine variety. It can be determined
as the dimension of the tangent space $T_f$ of the action of $\GG$ at $f$,
which is a vector space defined as follows.
Let $\ag := \IC^{w_1 w_2 \times w_1 w_2} \times \cdots \times \IC^{w_d w_{d+1} \times w_d w_{d+1}}$.
For $A \in \ag$ we define the \emph{Lie algebra action} $A f := \lim_{\varepsilon\to 0}\tfrac 1 \varepsilon\left((\textup{id}+\varepsilon A)f-f\right)$, where $\textup{id} \in \GG$ is the identity element.
We define the vector space
\[
T_f := \ag f = \{ A f \mid A \in \ag \}.
\]
\begin{claim}\label{cla:tangentspacedimensionscoincide}
The dimension $\dim T_{h}$ is
the same for all ${h} \in \GG f$.
\end{claim}
\begin{proof}
Since the action of $\GG$ is linear, for all $g \in \GG$ and $A \in \ag$ we have
\begin{eqnarray*}
A(gf) &=& \lim_{\varepsilon\to 0}\tfrac 1 \varepsilon\left((\textup{id}+\varepsilon A)(gf)-gf\right) = \lim_{\varepsilon\to 0}\tfrac 1 \varepsilon\left(g g^{-1}(\textup{id}+\varepsilon A)g f-gf\right) 
\\ &=& g \lim_{\varepsilon\to 0}\tfrac 1 \varepsilon\left((\textup{id}+\varepsilon (g^{-1}Ag)) f-f\right) = g((g^{-1}Ag)f)
\end{eqnarray*}
Since $A \mapsto g^{-1}Ag$ is a bijection on $\ag$, it follows that $T_{gf} = g T_f$.
Hence the claim follows.
\end{proof}
In the following we will use Claim~\ref{cla:tangentspacedimensionscoincide} to argue $f_0 \notin \GG\IMM$ by showing
that $\dim T_{\IMM}$ and $\dim T_{f_0}$ are different.

Let $e,e' \in E^i$ and let $A_{e,e'}^{(i)} \in \ag$ denote the matrix tuple where the $i$-th matrix has a 1
at position $(e,e')$ and all other entries (also in all other matrices) are 0.
Since these matrices form a basis of $\ag$, it follows that
\[
\ag f = \textup{linspan}\{ A_{e,e'}^{(i)} f \}.
\]
For a tensor $f$ we define the \emph{support} of $f$ as the set of monomials (i.e., standard basis tensors) for which $f$ has nonzero coefficient.
For a linear subspace $V \subseteq \IC^{w_1 w_2} \otimes \IC^{w_2 w_3} \otimes \cdots \otimes \IC^{w_d w_{d+1}}$ we define the \emph{support} of $V$ as the union of the supports of all $f \in V$.

We write $e \cap e' = \emptyset$ to indicate that two edges $e$ and $e'$ do not share any vertex.
We write $|e \cap e'| = 1$ if they share exactly one vertex.
We observe that for $f \in \{\IMM,f_0\}$
the vector space $T_f$ decomposes into a direct sum of three vector spaces,
\begin{eqnarray*}
\ag_2 &:=& \textup{linspan}\{ A_{e,e'}^{(i)} \mid 1 \leq i \leq d, 1 \leq e, e' \leq w_i w_{i+1}, e \cap e'=\emptyset \} \\
\ag_1 &:=& \textup{linspan}\{ A_{e,e'}^{(i)} \mid 1 \leq i \leq d, 1 \leq e, e' \leq w_i w_{i+1}, |e \cap e'|=1 \} \\
\ag_0 &:=& \textup{linspan}\{ A_{e,e}^{(i)} \mid 1 \leq i \leq d, 1 \leq e \leq w_i w_{i+1} \}.\\
\ag &=& \ag_0 \oplus \ag_1 \oplus \ag_2\\
T_f &=& \ag_0 f \oplus \ag_1 f \oplus \ag_2 f
\end{eqnarray*}
The last direct sum decomposition follows from the fact that $\ag_0 f$, $\ag_1 f$, and $\ag_2 f$ have pairwise disjoint supports.

We show in this section that $\dim \ag_2 \IMM = \dim \ag_2 f_0$, and that $\dim \ag_1 \IMM = \dim \ag_1 f_0$.
In Section~\ref{sec:flows} we show that
$\dim \ag_0 \IMM > \dim \ag_0 f_0$, which then implies $f_0 \notin \GG \IMM$ by Claim~\ref{cla:tangentspacedimensionscoincide}.
In fact, Theorem~\ref{thm:dimT} gives the exact dimension of $\ag_0 \IMM$ by proving that $\ag_0 \IMM$ is isomorphic to the vector space of flows on the ABP digraph when identifying vertices in $V^1$ with their corresponding vertices in $V^{d+1}$. Theorem~\ref{thm:dimTprime} establishes an additional equation based on the vertex parities that shows that $\ag_0 f_0$ is strictly lower dimensional than $\ag_0 \IMM$.

We start with Lemma~\ref{lem:left-summand-full-dim}, which shows that $\dim \ag_2 \IMM$ and $\dim \ag_2 f_0$ have \emph{full} dimension.
\begin{lemma}
  \label{lem:left-summand-full-dim}
  Let $f \in \{\IMM,f_0\}$. The space $\ag_2 f$ has full dimension. That is, its dimension equals \(\sum_{i=1}^d w_iw_{i+1}(w_i-1)(w_{i+1}-1)\).
\end{lemma}
\begin{proof}
  Suppose $f = \IMM$. The other case being analogous, we only argue this case.
  
  We analyze the monomials that appear in the different $A_{e,e'}^{(i)}\IMM$ and argue that a monomial that appears in some $A_{e,e'}^{(i)}\IMM$
  can only appear in that specific $A_{e,e'}^{(i)}\IMM$.
  Indeed, each monomial corresponds to a valid path in which one edge $e$ in layer $i$ is changed to $e'$.
  Since $e$ and $e'$ share no vertex, from this edge sequence we can reconstruct $i$, $e$, and $e'$ uniquely:
  $e'$ is the edge that does not have any vertex in common with the rest of the edge sequence,
  $i$ is its layer, and $e$ is the unique edge that we can replace $e'$ by in order to form a valid path.
  We conclude that the $A_{e,e'}^{(i)}\IMM$ have disjoint support and the lemma follows.
  \end{proof}

To establish that $\dim \ag_1 \IMM = \dim \ag_1 f_0$, we introduce some notation.

For a connected directed graph $G=(V,E)$ we define a \emph{flow} to be a labeling of the edge set $E$ by complex numbers such that at every vertex the sum of the labels of the incoming edges equals the sum of the labels of the outgoing edges. It is easily seen that the set of flows forms a vector space $F$.
We have
\begin{equation}\label{eq:flow}
\dim F = |E|-|V|+1,
\end{equation}
see e.g.\ Theorem~20.7 in~\cite{BM-book08}.

Recall that $E^i$ denotes the set of edges from $V^i$ to $V^{i+1}$. 
Let $\mathscr X := E^1 \times \cdots \times E^d$ denote the direct product of the sets of edge lists.
Each directed path of length $d$ from layer $1$ to $d+1$ is an element of $\mathscr X$, but $\mathscr X$ contains other edge sets as well.
Define $E_i := \IC^{E^i}$. Consider the following map $\varphi$ from
$\mathscr{X}$ to $E_1 \otimes \cdots \otimes E_d$, 
\[
\varphi (e_1,\ldots,e_d) = x_{e_1} \otimes \cdots \otimes x_{e_d} \in E_1 \otimes \cdots \otimes E_d
\]
where $(x_j)$ is the standard basis of $E_i$. Note $\varphi$ is a bijection between $\mathscr X$ and the standard basis of $E_1 \otimes \cdots \otimes E_d$.

An edge set in $\mathscr X$ is called a \emph{valid path} if it forms a path that starts and ends at corresponding vertices (see Sec.~\ref{sec:intro}).
Let $\mathscr P \subseteq \mathscr X$ denote the set of valid paths.

\begin{proposition}\label{pro:gonefcom}
$\dim \ag_1 \IMM = \dim \ag_1 f_0 = \sum_{i=1}^d (w_{i-1} + w_{i+1} -1)(w_i-1)w_i$, where $w_0 := w_d$.
\end{proposition}
\begin{proof}
The proof works almost analogously for $\IMM$ and $f_0$, so we treat only the more natural case~$\IMM$.
We show that $\ag_1 \IMM$ is isomorphic to a direct sum of vector spaces of flows on very simple digraphs.
Fix $1 \leq i \leq d$. Fix distinct $1 \leq a,b \leq w_i$.
For distinct edges $e,e' \in E^i$, let $\mathscr P_{e,e'} \subseteq \mathscr X$ be the set of edge sets containing $e'$ that are not valid paths, but that become valid paths by removing $e'$ and adding $e$.
Let $\mathscr P_{a,b}^i \subseteq \mathscr X$ be the set of edge sets that are not valid paths, but that become valid paths by switching the end point of the $(i-1)$-th edge to $v^i_b$ and that also become valid paths by switching the start point of the $i$-th edge to $v^i_a$ (if $i-1=0$, then interpret $i-1:=d$). Pictorially, this means that elements in $\mathscr P_{a,b}^i$ are almost valid paths, but there is a discontinuity at layer $i$, where the path jumps from vertex $v^i_a$ to vertex $v^i_b$.
We have
\[
A^{(i)}_{e,e'}\IMM = \sum_{p \in \mathscr P_{e,e'}}\varphi(p).
\]
The vectors $\{A^{(i)}_{e,e'}\IMM \mid 1\leq i \leq d, e,e' \in E^i, |e\cap e'|=1\}$ are not linearly independent, because for $a \neq b$ we have
\begin{equation}\label{eq:equality}
\sum_{\substack{
\text{$e$ and $e'$ have the same start point}\\
e' \text{ ends at the $a$-th vertex}\\
e \text{ ends at the $b$-th vertex}
}}
A^{(i-1)}_{e,e'}\IMM
=
\sum_{p \in \mathscr P^i_{a,b}}\varphi(p)
=
\sum_{\substack{
\text{$h$ and $h'$ have the same end point}\\
h \text{ starts at the $a$-th vertex}\\
h' \text{ starts at the $b$-th vertex}
}}
A^{(i)}_{h,h'}\IMM.
\end{equation}
Define
\begin{eqnarray*}
T_{a,b,i} &:=& \textup{linspan}\bigg\{A^{(i-1)}_{e,e'}\IMM \ \bigg\lvert \ \substack{
\text{$e$ and $e'$ have the same start point}\\
e' \text{ ends at the $a$-th vertex}\\
e \text{ ends at the $b$-th vertex}
}\bigg\} \\
&+& \textup{linspan}\bigg\{A^{(i)}_{h,h'}\IMM \ \bigg\lvert \  \substack{
\text{$h$ and $h'$ have the same end point}\\
h \text{ starts at the $a$-th vertex}\\
h' \text{ starts at the $b$-th vertex}
}\bigg\}.
\end{eqnarray*}
The support of $T_{a,b,i}$ and $T_{\tilde a,\tilde b,\tilde i}$ are disjoint, provided $(a,b,i) \neq (\tilde a, \tilde b, \tilde i)$.
Hence
\[
\ag_1 \IMM = \bigoplus_{\substack{1 \leq i \leq d\\1 \leq a,b \leq w_i\\a \neq b}} T_{a,b,i}
\]
It remains to prove that the dimension of $T_{a,b,i}$ is $w_{i-1}+w_{i+1}-1$, because then
\[
\dim \ag_1 \IMM = \sum_{\substack{1 \leq i \leq d\\1 \leq a,b \leq w_i\\a \neq b}} (w_{i-1} + w_{i+1} -1) =
\sum_{i=1}^d (w_{i-1} + w_{i+1} -1)(w_i-1)w_i.
\]
Note that $T_{a,b,i}$ is defined as the linear span of $w_{i-1}+w_{i+1}$ many vectors, but \eqref{eq:equality} shows that these are not linearly independent.
We prove that \eqref{eq:equality} is the only equality by showing that $T_{a,b,i}$ is isomorphic to a flow vector space.
We define a multigraph with two vertices: $\vi$ and $\vstar$. We have $w_{i+1}$ many edges from $\vi$ to $\vstar$, and we have $w_{i-1}$ many edges from $\vstar$ to $\vi$.
We denote by $\vstar\stackrel{k}{\to}\vi$ the $k$-th edge from $\vstar$ to $\vi$.
Let $F_{a,b,i}$ denote the vector space of flows on this graph. Its dimension is $w_{i-1}+w_{i+1}-1$, see \eqref{eq:flow}.
We define $\varrho : E^1 \otimes \cdots \otimes E^d \to F_{a,b,i}$ on rank 1 tensors via
\begin{eqnarray*}
\varrho(x_{e_1}\otimes \cdots \otimes x_{e_d})(\vstar\stackrel{k}{\to}\vi) &=& \begin{cases}
1 & \text{ if $e_{i-1}$ starts at $k$ in layer $i-1$ and ends at $a$ in layer $i$,} \\
0 & \text{ otherwise}.
\end{cases} \\
\varrho(x_{e_1}\otimes \cdots \otimes x_{e_d})(\vi\stackrel{l}{\to}\vstar) &=& \begin{cases}
1 & \text{ if $e_{i}$ starts at $b$ in layer $i$ and ends at $l$ in layer $i+1$,} \\
0 & \text{ otherwise}.
\end{cases}
\end{eqnarray*}
Using \eqref{eq:equality} it is readily verified that $\varrho$ maps $T_{a,b,i}$ to $F_{a,b,i}$. It remains to show that $\varrho : T_{a,b,i} \to F_{a,b,i}$ is surjective.
Let $\alpha := |\mathscr P_{a,b}^i|$.
We observe that
\begin{eqnarray*}
\varrho(A_{e,e'}^{(i-1)}\IMM)(\vstar\stackrel{k}{\to}\vi) &=&
\begin{cases}
\alpha/w_{i-1} & \text{if $e$ and $e'$ both start at the $k$-th vertex} \\
0 & \text{if $e$ and $e'$ both start at the same vertex, but not at the $k$-th}
\end{cases} \\
\varrho(A_{e,e'}^{(i-1)}\IMM)(\vi\stackrel{l}{\to}\vstar) &=& \alpha/(w_{i-1}w_{i+1})\\
\varrho(A_{h,h'}^{(i)}\IMM)(\vi\stackrel{l}{\to}\vstar) &=&
\begin{cases}
\alpha/w_{i+1} & \text{if $h$ and $h'$ both end at the $l$-th vertex} \\
0 & \text{if $h$ and $h'$ both end at the same vertex, but not at the $l$-th}
\end{cases} \\
\varrho(A_{h,h'}^{(i)}\IMM)(\vstar\stackrel{k}{\to}\vi) &=& \alpha/(w_{i-1}w_{i+1})
\end{eqnarray*}
Let $\Xi := \sum A_{e,e'}^{(i-1)}\IMM$.
Then $\forall k: \varrho(\Xi)(\vstar\stackrel{k}{\to}\vi)=\alpha/w_{i-1}$ and
$\forall l: \varrho(\Xi)(\vi\stackrel{l}{\to}\vstar)=\alpha$.
Therefore, for $e,e'$ starting at the $k_0$-th vertex and $h,h'$ ending at the $l_0$-th vertex we have that
\[
\varrho\bigg(w_{i-1}w_{i+1}\varrho(A_{e,e'}^{(i-1)}\IMM) + w_{i-1}w_{i+1}\varrho(A_{h,h'}^{i}\IMM) - \Xi\bigg)
\]
is nonzero only on exactly two edges: $\vstar\stackrel{k_0}{\to}\vi$ and $\vi\stackrel{l_0}{\to}\vstar$. Cycles form a generating set of the vector space $F_{a,b,i}$, which finishes the proof of the surjectivity of $\varrho$.
\end{proof}

\section{Flows on ABPs}\label{sec:flows}
We now proceed to the analysis of $\ag_0 \IMM$ and $\ag_0 f_0$. The connection to flow vector spaces will be even more prevalent than in Proposition~\ref{pro:gonefcom}.
The main result of this section is $\dim \ag_0 \IMM > \dim \ag_0 f_0$ (Theorems~\ref{thm:dimT} and \ref{thm:dimTprime}), which implies that $\IMM$ and $f_0$ have different orbit dimensions. We thereby conclude that $f_0 \notin \GG \IMM$.

To each edge $e$ we assign its \emph{path tensor} $\psi(e)$ by summing tensors over all valid paths passing through $e$,
\[
\psi(e) := \sum_{p \in \mathscr P \text{ with } e \in p} \varphi(p) \in E_1 \otimes \cdots \otimes E_d.
\]
By linear continuation this gives a linear map
\(
\psi : \IC^E \to E_1 \otimes \cdots \otimes E_d.
\)
Observe that 
\(
\psi(e) = A_{e,e}^{(i)} \IMM.
\)
Let $\mathscr T$ denote the linear span of all $\psi(e)$, $e \in E$.
In other words, $\mathscr T = \ag_0 \IMM$.

Let $\mathscr P' \subseteq \mathscr P \subseteq \mathscr X$ be the set of valid paths that contain exactly one parity preserving edge.  
To each edge $e$ we assign its \emph{parity path tensor} $\psi'(e)$ by summing tensors over paths in $\mathscr{P'}$, 
\[
\psi'(e) := \sum_{p \in \mathscr P' \text{ with } e \in p} \varphi(p) \in E_1 \otimes \cdots \otimes E_d.
\]
By linear continuation this gives a linear map
\(
\psi' : \IC^E \to E_1 \otimes \cdots \otimes E_d.
\)
Observe that 
\(
\psi'(e) = A_{e,e}^{(i)} f_0.
\)
Let $\mathscr T'$ denote the linear span of all $\psi'(e)$, $e \in E$.
In other words, $\mathscr T' = \ag_0 f_0$.

We will establish the following bounds on the dimensions of $\mathscr T$ and $\mathscr T'$.

\begin{theorem}\label{thm:dimT}
$\dim \mathscr T = |E|-\sum_{i=1}^dw_i +  1$.
\end{theorem}

\begin{theorem}\label{thm:dimTprime}
$\dim \mathscr T' \leq |E|- \sum_{i=1}^dw_i $.
\end{theorem}

The rest of this section is dedicated to the proofs of Theorem~\ref{thm:dimT} and Theorem~\ref{thm:dimTprime} by showing that $\mathscr{T}$ is isomorphic to the vector space of flows ``on the ABP'', while the parity constraints lead to a smaller dimension of $\mathscr{T'}$.

From an ABP $\Gamma$ we construct a digraph $\tilde\Gamma$ by identifying corresponding vertices from the first and the last layer in $V$ and calling the resulting vertex set $\tilde V$.
Note $|\tilde V| = \sum_{i=1}^d w_i$.
The directed graphs $\Gamma$ and $\tilde\Gamma$ have the same edge set.
The resulting directed graph is called $\tilde \Gamma = (\tilde V,E)$.
Let $F$ denote the vector space of flows on $\tilde\Gamma$. Note that by \eqref{eq:flow} we have $\dim F = |E|- |\tilde V| +1$.
All directed cycles in $\tilde \Gamma$ have a length that is a multiple of~$d$. In particular, all cycles of length exactly $d$ are in one-to-one correspondence with valid paths in $\Gamma_{\textup{com}}$.
For an edge $e \in E$, let $\chi(e) \in \IC^E$ denote the characteristic function of $e$,
i.e., the function whose value is 1 on $e$ and 0 everywhere else.

 We now prove Theorem~\ref{thm:dimT} by establishing a matching upper (Lemma~\ref{lem:dimTleqF}) and lower bound (Lemma~\ref{lem:dimTgeqF}) of $|E|- |\tilde V| +1=\dim F$ on $\dim \mathscr{T}$.



\subsection*{The upper bound}

\begin{lemma}\label{lem:dimTleqF}
$\dim \mathscr T \leq |E|-|\tilde V|+1$.
\end{lemma}
\begin{proof}
For $v \in \tilde V$, 
let $\Ein(v) \subseteq E$ denote the set of incoming edges incident to $v$ and $\Eout(v) \subseteq E$ denote the set of outgoing edges incident to $v$. 
For each $v \in \tilde V$, define the row vector 
\[r_v = 
\sum_{e \in \Ein(v)} \chi(e) - \sum_{e \in \Eout(v)} \chi(e).
\]
These vectors are the rows of the signed incidence matrix of
$\tilde\Gamma$, and since $\tilde\Gamma$ is connected, they span a
space of dimension $|\tilde V|-1$ (\cite[Ex.~1.5.6]{BM-book08}).
Now observe that for all $v \in \tilde V$,
\[
\sum_{e \in \Ein(v)} \psi(e) = \sum_{e \in \Eout(v)} \psi(e).
\]
Since $\psi$ is linear, this is equivalent to
\[
\psi\left(\sum_{e \in \Ein(v)} \chi(e) - \sum_{e \in \Eout(v)} \chi(e)\right) = 0.
\]
Hence each $r_v$ is in the kernel of $\psi$, and hence $\dim \ker \psi
\geq |\tilde V|-1$. Using \eqref{eq:flow}, we obtain
$\dim \mathscr T = \dim \im \psi = |E|-\dim \ker \psi \leq |E|-|\tilde V|+1 = \dim F$.
\end{proof}

\subsection*{The lower bound}

To obtain the lower bound, we define a linear map \(\varrho : E_1 \otimes \cdots \otimes E_d \to \IC^E \) such that the image of $\varrho$ restricted to $\mathscr{T}$ equals $F$.  This will imply that $\dim \mathscr{T} \geq \dim F$, thereby achieving the required lower bound.  

We define the linear map $\varrho$ 
on standard basis elements $x_{e_1} \otimes \cdots \otimes x_{e_d}$
as follows, 
\[\varrho(x_{e_1} \otimes \cdots \otimes x_{e_d}) := \chi(e_1)+\cdots+\chi(e_d), \]
and then extend it to the domain $E_1 \otimes \cdots \otimes E_d$ via linear continuation.

\begin{lemma}\label{lem:dimTgeqF}
Let $\varrho|_{\mathscr T}$ denote the restriction of $\varrho$ to the linear subspace $\mathscr T$.
Then, $\im \varrho|_{\mathscr T} = F$.
In particular, $\dim \mathscr T \geq \dim F = |E|-|\tilde V|+1$.
\end{lemma}
\begin{proof}
  To prove equality it suffices to show  $\im \varrho|_{\mathscr T} \subseteq F$ and  $F \subseteq \im \varrho|_{\mathscr T}$.

  The first containment is easy to see. For an edge $e$, consider the image of $\psi(e)$ under the map~$\varrho$,
  \[\varrho(\psi(e)) = \sum_{e \in p \in \mathscr P} \sum_{e' \in p}\chi(e').\]
  Observe that for a path $p \in \mathscr P$,  $\sum_{e' \in p}\chi(e')$ is a flow on $\tilde\Gamma$ and hence it belongs to $F$. Thus, we have $\varrho(\psi(e)) \in F$. Since $\mathscr{T}$ is spanned by $\psi(e)$, for $e \in E$, we obtain that $\im \varrho|_{\mathscr T} \subseteq F$.

  To establish the second containment it suffices to show that the image of $\mathscr T$ under the map $\varrho$ contains a basis of $F$. We identify a specific basis for $F$ in Claim~\ref{cla:flowspan} and prove that it is contained in $\im \varrho|_{\mathscr T}$ in Claim~\ref{cla:construction} to complete the argument.  
\end{proof}

We identify directed cycles with their characteristic flows, i.e., flows that have value 1 on the cycle's edges and 0 everywhere else.
We also identify directed cycles that use edges in any direction with their characteristic flow: the characteristic flow is defined to take the value 1 on an edge $e$ if $e$ is traversed in the direction of $e$,
and value $-1$ on $e$ if $e$ is traversed against its direction.

From the theory of flows we know that for every (undirected) spanning tree $T$ of $\tilde\Gamma$, the vector space $F \in \IC^E$ has a basis given by the
characteristic flows of cycles that only use edges from $T$ and exactly one additional edge (for example, see Theorem~20.8 in~\cite{BM-book08}).
Thus, the cycle flows corresponding to the elements not in the spanning tree form a basis of~$F$.

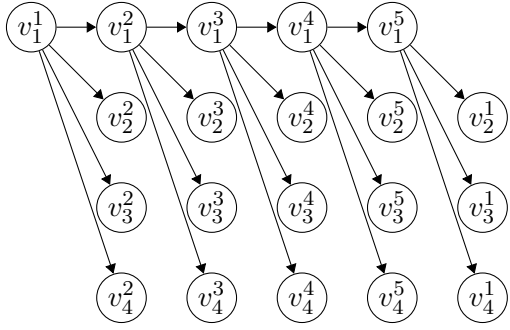
\begin{figure}
\begin{center}
\begin{tikzpicture}[scale=1.2]
\node[draw,circle,inner sep=1pt] (v11) at (0,0) {$v^1_1$};
%
\node[draw,circle,inner sep=1pt] (v21) at (1,0) {$v^2_1$};
\node[draw,circle,inner sep=1pt] (v22) at (1,-1) {$v^2_2$};
\node[draw,circle,inner sep=1pt] (v23) at (1,-2) {$v^2_3$};
\node[draw,circle,inner sep=1pt] (v24) at (1,-3) {$v^2_4$};
\node[draw,circle,inner sep=1pt] (v31) at (2,0) {$v^3_1$};
\node[draw,circle,inner sep=1pt] (v32) at (2,-1) {$v^3_2$};
\node[draw,circle,inner sep=1pt] (v33) at (2,-2) {$v^3_3$};
\node[draw,circle,inner sep=1pt] (v34) at (2,-3) {$v^3_4$};
\node[draw,circle,inner sep=1pt] (v41) at (3,0) {$v^4_1$};
\node[draw,circle,inner sep=1pt] (v42) at (3,-1) {$v^4_2$};
\node[draw,circle,inner sep=1pt] (v43) at (3,-2) {$v^4_3$};
\node[draw,circle,inner sep=1pt] (v44) at (3,-3) {$v^4_4$};
\node[draw,circle,inner sep=1pt] (v51) at (4,0) {$v^5_1$};
\node[draw,circle,inner sep=1pt] (v52) at (4,-1) {$v^5_2$};
\node[draw,circle,inner sep=1pt] (v53) at (4,-2) {$v^5_3$};
\node[draw,circle,inner sep=1pt] (v54) at (4,-3) {$v^5_4$};
%
\node[draw,circle,inner sep=1pt] (v62) at (5,-1) {$v^1_2$};
\node[draw,circle,inner sep=1pt] (v63) at (5,-2) {$v^1_3$};
\node[draw,circle,inner sep=1pt] (v64) at (5,-3) {$v^1_4$};
\draw[-Triangle] (v11) -- (v21);
\draw[-Triangle] (v11) -- (v22);
\draw[-Triangle] (v11) -- (v23);
\draw[-Triangle] (v11) -- (v24);
\draw[-Triangle] (v21) -- (v31);
\draw[-Triangle] (v21) -- (v32);
\draw[-Triangle] (v21) -- (v33);
\draw[-Triangle] (v21) -- (v34);
\draw[-Triangle] (v31) -- (v41);
\draw[-Triangle] (v31) -- (v42);
\draw[-Triangle] (v31) -- (v43);
\draw[-Triangle] (v31) -- (v44);
\draw[-Triangle] (v41) -- (v51);
\draw[-Triangle] (v41) -- (v52);
\draw[-Triangle] (v41) -- (v53);
\draw[-Triangle] (v41) -- (v54);
\draw[-Triangle] (v51) -- (v62);
\draw[-Triangle] (v51) -- (v63);
\draw[-Triangle] (v51) -- (v64);
\end{tikzpicture}
\end{center}
\caption{The spanning tree construction for width $4$ and $d=5$.}
\label{fig:spanningtree}
\end{figure}

\begin{claim}\label{cla:flowspan}
$F$ is spanned by the set of directed cycles in $\tilde \Gamma$ of length exactly $d$.
\end{claim}
\begin{proof}
We construct a spanning tree $\tau$ as follows,
which will be a tree whose edges are all directed away from its root.
Informally, the tree is given by the following subgraph, we make the first vertex in $V^1$ as root, and include all the outgoing edges incident to it. We then move to the first vertex in $V^2$ and include all the outgoing edges incident to it. We continue in this way until we reach $V^d$. Upon reaching the first vertex in $V^d$ we include all but one outgoing edges incident to it. The one that is an incoming edge to the root is not included. Figure~\ref{fig:spanningtree} illustrates the construction. We now formally define this.

Let $v^i_1 \in V^i$ denote the first vertex in the layer $i$, $1 \leq i \leq d$. Further recall $\Ein(v) \subseteq E$ and $\Eout(v) \subseteq E$ denote the set of incoming and outgoing edges, respectively, incident to $v$. Define the edge set
\[
\tau:= \left(\bigcup_{i=1}^d \Eout(v_1^i)\right) \setminus \{(v_1^d,v_1^1)\},
\]
which is a spanning tree in $\tilde \Gamma$. We know that every edge not in
the tree when added to the tree gives a unique undirected cycle. We now show
that the characteristic flows of these undirected cycles can be expressed
as a linear combination of the characteristic flows of directed cycles of length $d$.
For $e \in E \setminus \tau$, let $c_e$ denote the characteristic flow of
the unique undirected cycle that uses $e$ in its correct direction and only edges of $\tau$.
We argue depending on which layer the edge $e$ belongs to.

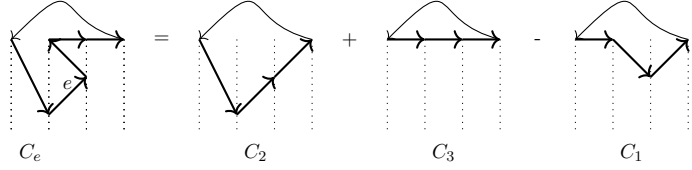
\begin{figure}
\begin{center}
\begin{tikzpicture}[scale=0.5, every node/.style={scale=0.7}]

    \coordinate (up) at (0,1) {};
    \coordinate (down) at (0,-1) {};
    \coordinate (left) at (-1,0) {};
    \coordinate (right) at (1,0) {};

    \coordinate (Ce) at (1,0) {};
    \coordinate (C2) at (6,0) {};
    \coordinate (C3) at (11,0) {};
    \coordinate (C1) at (16,0) {};

    \path [decoration=arrows, decorate] (Ce) -- ++($2*(down)+(right)$) -- node[above] {$e$}
    ++($(up)+(right)$) -- ++($(up)+(left)$) --
    ++($(right)$) --     ++($(right)$); 
    \draw[->] ($(Ce)+3*(right)$) .. controls ($(Ce)+(right)+0.5*(up)$) and
    ($(Ce)+2*(right)+2*(up)$)    .. (Ce);  
    \node[left] (labelCe) at ($(Ce)+3*(down)+(right)$) {$C_e$};
    \foreach \i in {0,1,2,3}
      \draw[dotted] ($(Ce)+(\i,0)$) -- ++($2.5*(down)$);
    
    \node at ($(C2) + (left)$) {=};
    \node at ($(C3) + (left)$) {+};
    \node at ($(C1) + (left)$) {-};
    
    \path [decoration=arrows, decorate] (C2) -- ++($2*(down)+(right)$) -- 
    ++($(up)+(right)$) -- ++($(up)+(right)$); 
    \draw[->] ($(C2)+3*(right)$) .. controls ($(C2)+(right)+0.5*(up)$) and
    ($(C2)+2*(right)+2*(up)$)    .. (C2);  
    \node[right] (labelC2) at ($(C2)+3*(down)+(right)$) {$C_2$};

    \path [decoration=arrows, decorate] (C3) -- ++($(right)$) -- 
    ++($(right)$) -- ++($(right)$); 
    \draw[->] ($(C3)+3*(right)$) .. controls ($(C3)+(right)+0.5*(up)$) and
    ($(C3)+2*(right)+2*(up)$)    .. (C3);  
    \node[right] (labelC3) at ($(C3)+3*(down)+(right)$) {$C_3$};
    \foreach \i in {0,1,2,3}
      \draw[dotted] ($(Ce)+(\i,0)$) -- ++($2.5*(down)$);

    \path [decoration=arrows, decorate] (C1) -- ++($(right)$) -- 
    ++($(down)+(right)$) -- ++($(up)+(right)$); 
    \draw[->] ($(C1)+3*(right)$) .. controls ($(C1)+(right)+0.5*(up)$) and
    ($(C1)+2*(right)+2*(up)$)    .. (C1);  
    \node[right] (labelC1) at ($(C1)+3*(down)+(right)$) {$C_1$};

    \foreach \i in {0,1,2,3}
        \foreach \j in {Ce,C1,C2,C3}
      \draw[dotted] ($(\j)+(\i,0)$) -- ++($2.5*(down)$);
\end{tikzpicture}
\end{center}
\caption{Decomposing a cycle of length $d+2$ as a linear combination of cycles of length $d$. The figure is an illustration when $d = 3$. The dotted layers in each cycle from the left are $V^3$, $V^1$, $V^2$, and $V^3$ again.}
\label{fig:linear-combination-of-cycles}
\end{figure}

\begin{itemize}
\item Suppose $e \in E^1 \setminus \tau$.
  \begin{itemize}
  \item If $e$ is incident to $v_1^2$, the first vertex in $V^2$, then the inclusion of $e$ creates a directed cycle of length $d$. Hence, $c_e$ equals the characteristic flow of this directed cycle.
  \item Otherwise, the inclusion of $e$ creates an undirected cycle of length $d+2$. If $e = (v^1_{j_1}, v^2_{j_2})$ for some $j_1 \in [2,w_1]$ and $j_2 \in [2,w_2]$, then the cycle $c_e$ is given as follows: \[v^d_1 - v^1_{j_1} - v^2_{j_2} - v^1_1 - v^2_1 - \cdots - v^{d-1}_1 - v^d_1.\]
    Consider the following two directed cycles:
    \begin{align*}
      C_1 ~\colon &  ~v^1_1 - v^2_{j_2} - \cdots - v^d_1 - v^1_1 \mbox { and } \\
      C_2 ~\colon &  ~v^1_{j_1} - v^2_{j_2} - \cdots - v^d_1 - v^1_{j_1},
    \end{align*}
    such that the part $ v^2_{j_2} - \cdots - v^d_1$ between $v^2_{j_2}$ and $v^d_1$ in the two cycles is the same. 
    Let us denote the characteristic flow of a cycle $C$ by $\chi(C)$. We now observe that $\chi(C_2) - \chi(C_1)$ equals the characteristic flow of the undirected cycle $v^1_{j_1} - v^2_{j_2} - v^1_1 - v^d_1 - v^1_{j_1}$. This is because the common part in $C_1$ and $C_2$ cancels out.
    To $\chi(C_2) - \chi(C_1)$ we add the characteristic flow of
    the directed cycle,
    \[C_3 ~\colon~ v^1_1 - v^2_1 - v^3_1 - \cdots - v^{d-1}_1 - v^d_1 - v^1_1 .\] 
    It is now easily seen that $\chi(C_2) - \chi(C_1) + \chi(C_3)$ equals the characteristic flow of the cycle $c_e$ (see Figure~\ref{fig:linear-combination-of-cycles} for an illustration). 
  \end{itemize}

\item Suppose $e \in E^d \setminus \tau$.
  \begin{itemize}
  \item If $e$ is incident to $v^1_1$, the first vertex in $V^1$, then as before the inclusion of $e$ creates a directed cycle of length $d$. Hence, $c_e$ equals the characteristic flow of this directed cycle.
  \item Otherwise, the inclusion of $e$ creates an undirected cycle of length 4. If $e = (v^d_{j_1}, v^1_{j_2})$ for some $j_1 \in [2,w_d]$ and $j_2 \in [2,w_1]$, then the cycle $c_e$ is given as follows:
    \[v^d_{j_1} - v^1_{j_2} - v^d_1 - v^{d-1}_1 - v^d_{j_1} .\]
    Consider the following two directed cycles:
    \begin{align*}
      C_4 ~\colon &  ~v^1_{j_2} - \cdots - v^{d-1}_1 - v^d_1 - v^1_{j_2} \mbox { and } \\
      C_5 ~\colon &  ~v^1_{j_2} - \cdots - v^{d-1}_1 -  v^d_{j_1} - v^1_{j_2},
    \end{align*}
    such that the part $ v^1_{j_2} - \cdots - v^{d-1}_1$ between $v^1_{j_2}$ and $v^{d-1}_1$ in the two cycles is the same.
    We now claim that $\chi(C_5) - \chi(C_4)$ equals the characteristic flow of $c_e$. This is because the common part in $C_4$ and $C_5$ cancels out. 
  \end{itemize}

\item Otherwise $e \in E^i \setminus \tau$ for some $i \in \{2,\ldots , d-1\}$. In such a case inclusion of $e$ creates an undirected cycle of length 4. We can again argue exactly like in the previous case, and so we omit the argument here. \qedhere
\end{itemize}
\end{proof}

We now prove that the generating set given by the directed cycles of length $d$ is contained in the image of $\mathscr{T}$ under the map $\varrho$. 

\begin{claim}\label{cla:construction}
$\im(\varrho|_{\mathscr T})$ contains the characteristic flow of each directed cycle of length $d$.
\end{claim}
\begin{proof}
Let $\{e_1,e_2,\ldots,e_d\} \subseteq E$ be a directed cycle of length $d$, where each $e_i$ points from a vertex in $V^i$ to a vertex in $V^{i+1}$.
Let $\{e_i^{(j)}\}$ denote the set of edges that start at the same vertex as $e_i$, but for which $e_i^{(j)} \neq e_i$.
Thus $|\{e_i^{(j)}\}| = |V^{i+1}|-1$.
Let \[
\bar\psi(e) := \frac 1 {|\{p \in \mathscr P \text{ with } e \in p\}|} \psi(e),
\]
so that $\varrho(\bar\psi(e))$ is a flow with value 1 on the edge $e$.
It is instructive to have a look at the left side of Figure~\ref{fig:flowI}, where
$\varrho(\bar\psi(e_1))$ is depicted.
Subtracting $\tfrac{1}{w_3} \sum_{j=1}^{w_3-1}\varrho(\bar\psi(e_2^{(j)}))$ and adding $\tfrac{w_3-1}{w_3} \varrho(\bar\psi(e_2))$ reduces the support significantly and brings us one step closer to the cycle, see the right side of Figure~\ref{fig:flowI}. We iterate this process until only the cycle is left.
\begin{figure}
\begin{tikzpicture}[scale=1.2]
\node[draw,circle,inner sep=3pt] (v11) at (0,0) {};
\node[draw,circle,inner sep=3pt] (v12) at (0,-1) {};
\node[draw,circle,inner sep=3pt] (v13) at (0,-2) {};
\node[draw,circle,inner sep=3pt] (v14) at (0,-3) {};
\node[draw,circle,inner sep=3pt] (v21) at (1,0) {};
\node[draw,circle,inner sep=3pt] (v22) at (1,-1) {};
\node[draw,circle,inner sep=3pt] (v23) at (1,-2) {};
\node[draw,circle,inner sep=3pt] (v24) at (1,-3) {};
\node[draw,circle,inner sep=3pt] (v31) at (2,0) {};
\node[draw,circle,inner sep=3pt] (v32) at (2,-1) {};
\node[draw,circle,inner sep=3pt] (v33) at (2,-2) {};
\node[draw,circle,inner sep=3pt] (v34) at (2,-3) {};
\node[draw,circle,inner sep=3pt] (v41) at (3,0) {};
\node[draw,circle,inner sep=3pt] (v42) at (3,-1) {};
\node[draw,circle,inner sep=3pt] (v43) at (3,-2) {};
\node[draw,circle,inner sep=3pt] (v44) at (3,-3) {};
\node[draw,circle,inner sep=3pt] (v51) at (4,0) {};
\node[draw,circle,inner sep=3pt] (v52) at (4,-1) {};
\node[draw,circle,inner sep=3pt] (v53) at (4,-2) {};
\node[draw,circle,inner sep=3pt] (v54) at (4,-3) {};
\node[draw,circle,inner sep=3pt] (v61) at (5,0) {};
\node[draw,circle,inner sep=3pt] (v62) at (5,-1) {};
\node[draw,circle,inner sep=3pt] (v63) at (5,-2) {};
\node[draw,circle,inner sep=3pt] (v64) at (5,-3) {};
\node at (0.5,0.5) {$\tfrac 1 {w_{4} w_5}$};
\node at (1.5,0.5) {$\tfrac 1 {w_5}$};
\node at (2.5,0.5) {$1$};
\node at (3.5,0.5) {$\tfrac 1 {w_3}$};
\node at (4.5,0.5) {$\tfrac 1 {w_3 w_4}$};
\draw[-Triangle] (v11) -- (v21);
\draw[-Triangle] (v11) -- (v22);
\draw[-Triangle] (v11) -- (v23);
\draw[-Triangle] (v11) -- (v24);
\draw[-Triangle] (v12) -- (v21);
\draw[-Triangle] (v12) -- (v22);
\draw[-Triangle] (v12) -- (v23);
\draw[-Triangle] (v12) -- (v24);
\draw[-Triangle] (v13) -- (v21);
\draw[-Triangle] (v13) -- (v22);
\draw[-Triangle] (v13) -- (v23);
\draw[-Triangle] (v13) -- (v24);
\draw[-Triangle] (v14) -- (v21);
\draw[-Triangle] (v14) -- (v22);
\draw[-Triangle] (v14) -- (v23);
\draw[-Triangle] (v14) -- (v24);
\draw[-Triangle] (v21) -- (v31);
\draw[-Triangle] (v22) -- (v31);
\draw[-Triangle] (v23) -- (v31);
\draw[-Triangle] (v24) -- (v31);
\draw[-Triangle] (v31) -- (v41);
\draw[-Triangle] (v41) -- (v51);
\draw[-Triangle] (v41) -- (v52);
\draw[-Triangle] (v41) -- (v53);
\draw[-Triangle] (v41) -- (v54);
\draw[-Triangle] (v51) -- (v61);
\draw[-Triangle] (v51) -- (v62);
\draw[-Triangle] (v51) -- (v63);
\draw[-Triangle] (v51) -- (v64);
\draw[-Triangle] (v52) -- (v61);
\draw[-Triangle] (v52) -- (v62);
\draw[-Triangle] (v52) -- (v63);
\draw[-Triangle] (v52) -- (v64);
\draw[-Triangle] (v53) -- (v61);
\draw[-Triangle] (v53) -- (v62);
\draw[-Triangle] (v53) -- (v63);
\draw[-Triangle] (v53) -- (v64);
\draw[-Triangle] (v54) -- (v61);
\draw[-Triangle] (v54) -- (v62);
\draw[-Triangle] (v54) -- (v63);
\draw[-Triangle] (v54) -- (v64);
\end{tikzpicture}\quad\quad\quad
\begin{tikzpicture}[scale=1.2]
\node[draw,circle,inner sep=3pt] (v11) at (0,0) {};
\node[draw,circle,inner sep=3pt] (v12) at (0,-1) {};
\node[draw,circle,inner sep=3pt] (v13) at (0,-2) {};
\node[draw,circle,inner sep=3pt] (v14) at (0,-3) {};
\node[draw,circle,inner sep=3pt] (v21) at (1,0) {};
\node[draw,circle,inner sep=3pt] (v22) at (1,-1) {};
\node[draw,circle,inner sep=3pt] (v23) at (1,-2) {};
\node[draw,circle,inner sep=3pt] (v24) at (1,-3) {};
\node[draw,circle,inner sep=3pt] (v31) at (2,0) {};
\node[draw,circle,inner sep=3pt] (v32) at (2,-1) {};
\node[draw,circle,inner sep=3pt] (v33) at (2,-2) {};
\node[draw,circle,inner sep=3pt] (v34) at (2,-3) {};
\node[draw,circle,inner sep=3pt] (v41) at (3,0) {};
\node[draw,circle,inner sep=3pt] (v42) at (3,-1) {};
\node[draw,circle,inner sep=3pt] (v43) at (3,-2) {};
\node[draw,circle,inner sep=3pt] (v44) at (3,-3) {};
\node[draw,circle,inner sep=3pt] (v51) at (4,0) {};
\node[draw,circle,inner sep=3pt] (v52) at (4,-1) {};
\node[draw,circle,inner sep=3pt] (v53) at (4,-2) {};
\node[draw,circle,inner sep=3pt] (v54) at (4,-3) {};
\node[draw,circle,inner sep=3pt] (v61) at (5,0) {};
\node[draw,circle,inner sep=3pt] (v62) at (5,-1) {};
\node[draw,circle,inner sep=3pt] (v63) at (5,-2) {};
\node[draw,circle,inner sep=3pt] (v64) at (5,-3) {};
\node at (0.5,0.5) {$\tfrac 1 {w_{4} w_5}$};
\node at (1.5,0.5) {$\tfrac 1 {w_5}$};
\node at (2.5,0.5) {$1$};
\node at (3.5,0.5) {$1$};
\node at (4.5,0.5) {$\tfrac 1 {w_4}$};
\draw[-Triangle] (v11) -- (v21);
\draw[-Triangle] (v11) -- (v22);
\draw[-Triangle] (v11) -- (v23);
\draw[-Triangle] (v11) -- (v24);
\draw[-Triangle] (v12) -- (v21);
\draw[-Triangle] (v12) -- (v22);
\draw[-Triangle] (v12) -- (v23);
\draw[-Triangle] (v12) -- (v24);
\draw[-Triangle] (v13) -- (v21);
\draw[-Triangle] (v13) -- (v22);
\draw[-Triangle] (v13) -- (v23);
\draw[-Triangle] (v13) -- (v24);
\draw[-Triangle] (v14) -- (v21);
\draw[-Triangle] (v14) -- (v22);
\draw[-Triangle] (v14) -- (v23);
\draw[-Triangle] (v14) -- (v24);
\draw[-Triangle] (v21) -- (v31);
\draw[-Triangle] (v22) -- (v31);
\draw[-Triangle] (v23) -- (v31);
\draw[-Triangle] (v24) -- (v31);
\draw[-Triangle] (v31) -- (v41);
\draw[-Triangle] (v41) -- (v51);
\draw[-Triangle] (v51) -- (v61);
\draw[-Triangle] (v51) -- (v62);
\draw[-Triangle] (v51) -- (v63);
\draw[-Triangle] (v51) -- (v64);
\end{tikzpicture}
\caption{On the left: $\varrho(\bar\psi(e_1))$.
On the right: $\varrho(\bar\psi(e_1))
- \tfrac{1}{w_3} \sum_{j=1}^{w_3-1}\varrho(\bar\psi(e_2^{(j)}))
+\tfrac{w_3-1}{w_3} \varrho(\bar\psi(e_2)) 
$.
This is the case $d=5$ and format $(4,4,4,4,4)$.
Edges that are not drawn carry 0 flow. All edges in the same layer carry either 0 flow or the value that is depicted above the edge layer.
For the purposes of illustation, $e_1$ is the top edge in the \emph{center}.
Here we assume that each $e_i$ points from the first vertex $V^i$ to the first vertex in $V^{i+1}$.}
 \label{fig:flowI}
\end{figure}
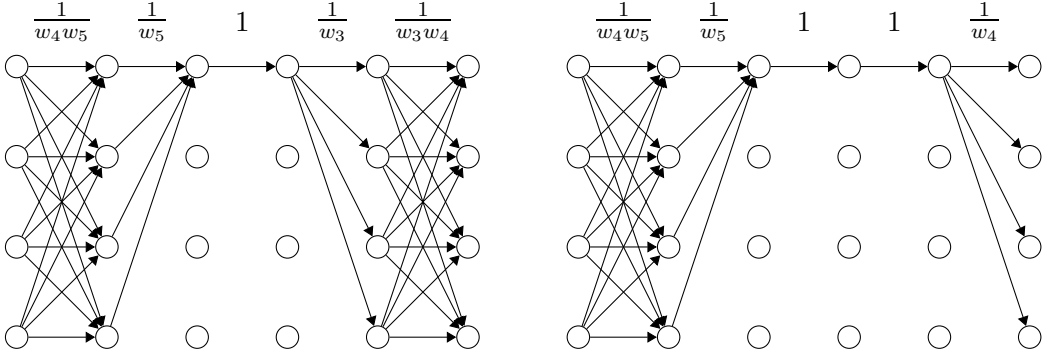
Formally:
\begin{eqnarray*}
\chi(e_1,\ldots,e_d) &=& \varrho(\bar\psi(e_1)) \\
&+& \tfrac{w_3-1}{w_3} \varrho(\bar\psi(e_2)) - \tfrac{1}{w_3} \sum_{j=1}^{w_3-1}\varrho(\bar\psi(e_2^{(j)})) \\
&+& \cdots \\
&+& \tfrac{w_d-1}{w_d} \varrho(\bar\psi(e_{d-1})) - \tfrac{1}{w_d} \sum_{j=1}^{w_d-1}\varrho(\bar\psi(e_{d-1}^{(j)})).
\end{eqnarray*}
\qedhere
\end{proof}


\subsection*{The stronger upper bound via parities}

We now proceed to upper bound $\dim \mathscr{T'}$ (Theorem~\ref{thm:dimTprime}).  The proof is analogous to the proof of Lemma~\ref{lem:dimTleqF}.    

\begin{theorem}[Restatement of Theorem~\ref{thm:dimTprime}]
$\dim \mathscr T' \leq |E| - |\tilde V|$.
\end{theorem}
\begin{proof}
  As in the proof of Lemma~\ref{lem:dimTleqF},
  for $v \in \tilde V$, we have 
\[
\sum_{e \in \Ein(v)} \psi'(e) = \sum_{e \in \Eout(v)} \psi'(e). 
\]
Furthermore, we have the following additional constraint on $\psi'$, 
\[
(d-1) \sum_{e \text{ parity preserving} } \psi'(e) =  \sum_{e \text{ parity changing} } \psi'(e).
\]
By the linearity of $\psi'$, we have
\[
\psi'\left((d-1)\sum_{e \text{ parity preserving}} \chi(e)  - \sum_{e \text{ parity changing}} \chi(e)\right) = 0. 
\]
Therefore, the kernel of $\psi'$ is spanned by the vectors $(\sum_{e \in \Ein(v)} \chi(e) - \sum_{e \in \Eout(v)} \chi(e))$, for $v \in \tilde V$,
and an additional vector
$((d-1)\sum_{e \text{ parity preserving}} \chi(e)  - \sum_{e \text{ parity changing}} \chi(e))$.

We now claim that the new vector is linearly independent from the earlier set of vectors. We prove the claim by constructing a vector in $\IC^E$ that is orthogonal to the earlier set of vectors but is non-orthogonal to the additional vector. One such vector is given by the characteristic flow of the directed cycle \(v^1_1 - v^2_1 - v^3_1 - \cdots - v^{d-1}_1 - v^d_1 - v^1_1\). 

Thus, it follows that $\dim \ker \psi' \geq |\tilde V|$, and hence $\dim \mathscr{T'} \leq |E| - |\tilde V|$.  
\end{proof}

In the next section we continue our investigation of comparing exact complexity classes with the approximative complexity classes. This would be a comparison between two well known classes, namely $\VQP$ and $\overline{\VNP}$.

\section{\texorpdfstring{$\VQP$}{VQP} versus \texorpdfstring{$\overline{\VNP}$}{VNPbar}} \label{sec:vqp}
In this section, we compare the complexity classes $\mathrm{VQP}$ and $\overline{\mathrm{VNP}}$. Valiant in his seminal paper \cite{valiantcomplete} defined the complexity classes that are now called as $\VP$ and $\VNP$, and the central question of algebraic complexity is to understand whether the two complexity classes are indeed different as sets (Valiant's hypothesis). 
B\"urgisser \cite{Burgisser2000} defined the complexity class $\VQP$ and related it to the complexity classes $\VP$ and $\VNP$. We proceed to define the above three classes for establishing the context. For an exhaustive treatment of the classes, we refer the readers to B\"{u}rgisser's monograph \cite{Burgisser2000} from where we are lifting the definitions. 
We first need to define so-called p-families.
\begin{definition}
A sequence $f = (f_n)$ of multivariate polynomials over a field $k$ is called a $p$-family (over $k$) iff the number of variables as well as the degree of $f_n$ are bounded by polynomial functions in $n$.
\end{definition}
We now need to define the model of computation and the notion of complexity in order to define the complexity classes of interest.

\begin{definition}
A \emph{straight-line program} $\Gamma$ (expecting $m$ inputs) represents a sequence $(\Gamma_1, \ldots, \Gamma_r)$ of instructions $\Gamma_{\rho} = (\omega_{\rho}; i_{\rho}, j_{\rho})$ with operation symbols $\omega_{\rho} \in \{ +, - , * \}$ and the address $i_{\rho}, j_{\rho}$ which are integers satisfying $-m < i_{\rho}, j_{\rho} < \rho$. We call $r$ the $size$ of $\Gamma$.
\end{definition}
So, essentially, in a straight-line program, we either perform addition or subtraction or multiplication on the inputs or the previously computed elements. The size of the straight-line program naturally induces a size complexity measure on polynomials as follows:
\begin{definition}
The \emph{complexity} $L(f)$ of a polynomial $f \in \IIF[x_1, \ldots, x_n]$ is the minimal size of a straight-line program computing $f$ from variables $x_i$ and constants in $\IIF$.
\end{definition}

We are now all set to define the above discussed complexity classes.
\begin{definition}
A $p$-family $f = (f_n)$ is said to be $p$-computable iff the complexity  $L(f_n)$ is a polynomially bounded function of $n$. $\VP_{\IIF}$  consists of all $p$-computable families over the field~$\IIF$.
\end{definition}

\begin{definition}
A $p$-family $f = (f_n)$ is said to be $p$-\emph{definable} iff there exists a $p$-\emph{computable} family $g = (g_n)$, $g_n \in \IIF[x_1, \ldots, x_{u(n)}]$, such that for all $n$
\begin{equation*}
f_n(x_1, \ldots, x_{v(n)}) = \sum\limits_{e \in \{0,1 \}^{u(n) - v(n)}} g_n(x_1, \ldots, x_{v(n)}, e_{v(n) + 1}, \ldots, e_{u(n)}).
\end{equation*} 
The set of $p$-definable families over $\IIF$ forms the complexity class $\VNP_{\IIF}$.
\end{definition}

\begin{definition}
A $p$-family $f = (f_n)$ is said to be $qp$-\emph{computable} iff the complexity $L(f_{n})$ is a quasi-polynomially bounded function of $n$. The complexity class $\VQP_{\IIF}$ consists of all $qp$-computable families over $\IIF$.
\end{definition}
In the above three definitions, if the underlying field is clear from the context, we can drop the subscript $\IIF$ and simply represent the classes as $\VP, \VNP$ and $\VQP$ respectively. In what follows, the underlying field is always assumed to be $\IQ$, the field of rational numbers.

In \cite{Burgisser2000}, B\"urgisser  showed the completeness of the determinant polynomial for $\VQP$ under $qp$-projections and strengthened Valiant's hypothesis of $\VNP \not\subseteq \VP$ to $\VNP \not\subseteq \VQP$ and called it \emph{Valiant's extended hypothesis} (see \cite{Burgisser2000}, section 2.5). He also established that $\VP \subsetneq \VQP$ and went on to show that $\VQP \not\subseteq \VNP$ (see \cite{Burgisser2000}, Proposition 8.5 and Corollary 8.9). The main observation of this section is that his proof is stronger and is sufficient to conclude that $\VQP$ is not contained in the closure of $\VNP$ either, where the closure is in the sense as mentioned in Section \ref{sec:intro}.

In fact, B\"urgisser in his monograph \cite{Burgisser2000} also gives a set of conditions which if the coefficients of a polynomial sequence satisfies, then that polynomial sequence cannot be in $\VNP$ (\cite{Burgisser2000}, Theorem 8.1). His theorem and the proof is inspired by Heintz and Sieveking \cite{heintz1980lower}. The second observation of this section is that this proof is even stronger and actually those conditions are sufficient to show that the given polynomial sequence is not contained in $\overline{\VNP}$ either.

We now discuss both the observations.

\subsection{\texorpdfstring{$\VQP \not\subseteq \overline{\VNP}$}{VQP not in VNPbar}}
We first show that there is a $\log n$ variate polynomial of degree $(n-1) \log n$ which is in $\VQP$ but not in $\overline{\VNP}$. In this exposition, for the sake of better readability, we do not present the B\"urgisser's statements in full generality since it is not essential for the theorem that we want to show here. Moreover, the less general version that we present here contains all the ideas for the theorem statements and their proofs.

\begin{theorem}\label{thm:vqpvnpbar}
Let $N_n := \{0, \ldots, n-1 \}^{\log n}$ and $f_n := \sum\limits_{\mu \in N_n}2^{2^{j(\mu)}}X_1^{\mu_1} \cdots X_{\log n}^{\mu_{\log n}}$, where $j(\mu) :=  \sum_{j = 1}^{\log n} \mu_{j}n^{j-1}$. 
Then $f_n \in \mathrm{VQP}$, but $f_n \notin \overline{\mathrm{VNP}}$, and hence $\VQP \not\subseteq \overline{\VNP}$.
\end{theorem}
The theorem consists of two parts. The containment in $\VQP$ follows immediately from the fact that the total number of monomials in $f_n$ is $n^{\log n}$.
For the other part, we closely follow B\"urgisser's lower bound proof (\cite{Burgisser2000}, Proposition 8.5) against $\VNP$ here, making transparent the fact that the proof works also against $\overline{\VNP}$. His proof techniques were borrowed from Strassen (\cite{strassen1974polynomials}). The idea is to use the universal representation for polynomial sequences in $\VNP$, so that we get a hold on how the coefficients of the polynomials look like. Using that, we establish polynomials $H_n$ that vanish on all the polynomial sequences in $\VNP$ (in other words, $H_n$ is in the vanishing ideal of sequences in $\VNP$), but do not vanish on $f_n$ (because the growth rate of its coefficients is too high), hence giving the separation. Since the vanishing ideal of a set characterizes its closure, we get the stronger separation, i.e., $f_n$ does not belong to the closure of $\VNP$, namely, $\overline{\VNP}$.

\begin{proof}[Proof of Theorem \ref{thm:vqpvnpbar}]
As stated above, the proof works in three stages: first, assuming the contrary and writing $f_n$ using the universal representation for the polynomial sequences in $\VNP$, then giving polynomials $H_n$ of special forms in the vanishing ideal of polynomial sequences in $\VNP$, and finally showing that $H_n$ cannot vanish on our sequence $f_n$, hence arriving at a contradiction.

Assuming $(f_n) \in \VNP$ implies the existence of a family $(g_n) \in \VP$, with $L(g_n)$ bounded by a polynomial $r(n)$, and a polynomial $u(n)$ such that \begin{equation*}
f_n(X_1, \ldots, X_{\log n}) = \sum\limits_{e \in \{0,1 \}^{u(n) - \log n}} g_n(X_1, \ldots, X_{\log n}, e_{\log n + 1}, \ldots, e_{u(n)}).
\end{equation*}

Next, we use the universal representation theorem (see \cite{strassen1974polynomials}, \cite{schnorr1977improved}) as stated in B\"urgisser's monograph (\cite{Burgisser2000}, Proposition 8.3; for a proof see \cite{burgisser2013algebraic}, Proposition 9.11) for size $r(n)$ straight-line program to get that there exist polynomials $G_{\nu}^{(n)} \in \IZ[Y_1, \ldots, Y_{q(n)}]$, 
with $q(n)$ being a polynomial in $n$ (more precisely, it is a polynomial in $r(n)$ and $u(n)$) which for $|\nu| \leq \deg g_n = n^{O(1)}$, guarantee that $\deg G_{\nu} = n^{ O(1)}, \log \wt(G_{\nu})^{(n)} = 2^{n^{O(1)}}$, and also guarantee the existence of some $\zeta \in \overline{\IQ}^{q(n)}$, such that \[g_n = \sum_{\nu}G_{\nu}^{(n)}(\zeta)X_1^{\nu_1}, \cdots, X_{u(n)}^{\nu_{u(n)}},\] where weight of a polynomial $f$, $\wt (f)$ refers to the sum of the absolute values of its coefficients. 

Now, taking exponential sum yields that \[f_n = \sum \limits_{\mu \in N_n} F_{\mu}^{(n)}(\zeta)X_1^{\mu_1} \cdots X_{\log n}^{\mu_{\log n}},\] where the polynomials $F_{\mu}^{(n)}$ are obtained as a sum of at most $2^{u(n)}$ polynomials $G_{\nu}^{(n)}$. 
Thus, we now have a good hold on $F_{\mu}^{(n)}$ i.e. $\deg F_{\mu}^{(n)} \leq \alpha(n)$ and $\log \wt(F_{\mu}^{(n)}) \leq 2 ^{\beta(n)}$, where both $\alpha(n)$ and $\beta(n)$ are polynomially bounded functions of $n$.

Thus, for $f_n$ to be in $\VNP$, the coefficients of $f_n$ should be in the image of the polynomial map $F_{\mu}^{n}: \overline{\IQ}^{q(n)} \rightarrow \overline{\IQ}^{n^{\log n}}$. In other words, we must have some $\zeta \in  \overline{\IQ}^{q(n)}$, such that for all $\mu \in N_n$, 
we have $F^{n}_{\mu}(\zeta) = 2^{2^{j(\mu)}}$, where $j(\mu) :=  \sum_{j = 1}^{\log n} \mu_{j}n^{j-1}$.
Since $F^{n}_{\mu}$ takes all the values from $2^{2^0}$ to $2^{2^{n^{\log n} -1}}$, we have a subset of indices $\tilde{N}_n \subseteq N_n$ of size $s(n):= \floor{|N_n|/n} = \floor{n^{\log n}/n}$, such that for $\sigma \in \{0, 1,\ldots, s(n)-1 \}$ and a bijection $\delta : \{0, 1,\ldots, s(n)-1 \} \rightarrow \tilde{N}_n$ with $\sigma \mapsto \delta(\sigma)$, we have $F^{n}_{\delta(\sigma)} = 2^{2^{\sigma n +1}}$.

Now we can apply Lemma 9.28 from \cite{burgisser2013algebraic} which asserts that there will be polynomials of low height (ht) (the maximum of the absolute value of the coefficients) on which these coefficients shall vanish. More precisely, there exists non-zero forms $H_{n} \in \IZ[Y_{\mu}  \mid  \mu \in \tilde{N}_n]$ 
with ht$(H_n) \leq 3$, $\deg H_{n} \leq D(n)$, and such that $H_{n}(F^{n}_{\mu} \mid \mu \in N_n) = 0$, given that $D(n)^{s(n)- q(n) -2} > \alpha(n)^{q(n)}s(n)^{s(n)}2^{\beta(n)}$. 

It can be seen that $D(n) = 2^{n}-1$ satisfies the above inequality, since $\alpha(n), \beta(n)$ and $q(n)$ are polynomially bounded and $2^n$ grows much faster than $s(n) = \floor{n^{\log n}/n}$. 
This allows us to write $H_{n} = \sum_{e} \lambda_{e} \prod_{\mu \in \tilde{N}_n}Y_{\mu}^{e_{\mu}}$, where the absolute values of $\lambda_e$ are bounded by $3$. Since  $H_n$ vanishes on the subset of coefficients of $f_n$ i.e it vanishes on $F^n_{\delta(\sigma)} = 2^{2^{\sigma n +1}}$ with $\sigma \in \{0,1,\ldots, s(n)-1 \}$, we have \[0 = H_n(F^{n}_{\mu} \mid \mu \in \tilde{N}_n) = \sum_{e} \lambda_{e}\prod_{\sigma=0 }^{s(n)-1}2^{e_{\delta(\sigma)}2^{\sigma n +1}} = \sum_{e} \lambda_{e} \cdot  4^{\sum_{\sigma} e_{\delta(\sigma)} (2^{n})^{\sigma}}.\] 
The last sum is essentially a $4$-adic integer, since firstly, $|\lambda_e| \leq 3$, and secondly, all the exponents of $4$, that is, $\sum_{\sigma} e_{\delta(\sigma)} (2^{n})^{\sigma}$ are all distinct, as they can be seen as $2^n$-adic representation since $e_{\delta(\sigma)} < 2^n$. Thus $\lambda_e$ has to be zero for all $e$. Hence $H_n$ must be identically zero, which is a contradiction.
\end{proof}

\subsection{A criterion for non-membership in \texorpdfstring{$\overline{\VNP}$}{criterion}}

In this section, we discuss a criterion B\"urgisser presented in his monograph \cite{Burgisser2000} based on a proof due to Heintz and Sieveking  which gives a set of conditions that puts a $p$-family out of $\VNP$. We observe that those conditions if satisfied, in fact, put a given $p$-family out of $\overline{\VNP}$ as well.  

\begin{theorem}\label{thm:criterion}
  Let $(p_n)$ be a sequence of polynomials over $\overline{\rat}$ and let
  $N(n)$ denote the degree of the field extension generated by the coefficients
  of $p_n$ over $\rat$. Further suppose the following holds:
  \begin{enumerate}
  \item The map $n \mapsto \lceil \log N(n) \rceil$ is not $p$-bounded.
  \item For all $n$, there is a system $G_n$ of rational polynomials of
    degree at most $D(n)$ with finite zeroset, containing the coefficient
    system of $f_n$, and such that $n \mapsto \lceil \log D(n)\rceil$
    is $p$-bounded. 
  \end{enumerate}
  Then the family $(p_n) \not\in \overline{\VNP}$. 
\end{theorem}

Thus the above theorem shows that certain $p$-families with algebraic coefficients of high degree are not contained in $\overline{\VNP}$. We now give a simple example from \cite{Burgisser2000} to illustrate the theorem. 
\begin{example}
Consider the following multivariate family defined as 
\[ p_n = \sum_{e \in \{0,1 \}^n  \char`\\ 0}  \sqrt{p_{j(e)}} X^{e}, \]
where $j(e) = \sum_{s=1}^n e_s2^{s-1}$ and $p_j$ refers to the $j$-th prime number. Then using the above Theorem~\ref{thm:criterion}, we can conclude that $p_n \notin \overline{\VNP}$.  This is because the degree of field extension $N(n) = [\IQ(\sqrt{p_j} \mid 1 \leq j \leq 2^n):\IQ] = 2^{2^n - 1}$ (see for example \cite{burgisser2013algebraic}, Lemma 9.20), hence condition~1 above is satisfied. Condition~2 is also satisfied because the coefficients are the roots of the system $G_{n} = \{Z_j^2 - p_j \mid 1 \leq j < 2^n \}$, with $D(n) = 2$.
\end{example}

For a proof of the theorem, we refer the readers to [\cite{Burgisser2000},Theorem 8.1]. We point out that the proof in its original form already works. In his proof, he wanted to conclude that $f_n \notin \VNP$. However, along the way, he arrives at a contradiction to the assertion that $f_n$ is contained in the Zariski closure of $\VNP$, which is exactly what is now known as $\overline{\VNP}$. During the time of the original proof, the complexity class $\overline{\VNP}$ was not defined.

\section*{Acknowledgements}
We thank Michael Forbes for illuminating discussions and for telling us about his (correct) intuition concerning Nisan's result.
We thank the Simons Institute for the Theory of Computing (Berkeley),
Schloss Dagstuhl - Leibniz-Zentrum f\"ur Informatik (Dagstuhl),
and the International Centre for Theoretical Sciences (Bengaluru),
for hosting us during several phases of this research.

\bibliographystyle{alpha}
\bibliography{abps}

\end{document}